\documentclass[12pt,a4paper,reqno]{amsart}
\usepackage[utf8]{inputenc}
\usepackage[T1]{fontenc}
\usepackage{indentfirst}
\usepackage[english]{babel}
\usepackage{amsmath, bbm}
\usepackage{amsthm, dsfont}
\usepackage{mathtools}
\usepackage{bbm}
\usepackage{amssymb}
\usepackage{enumitem}
\usepackage[dvipsnames]{xcolor}
\usepackage[colorlinks=true,citecolor=DarkOrchid,linkcolor=NavyBlue]{hyperref}
\usepackage{tikz}
\usepackage{array}
\usepackage{pdflscape}
\usepackage{caption}
\usepackage{subcaption}
\usepackage{float}
\usepackage{url}
\usepackage{stmaryrd} 
\usepackage{algorithm2e} 
\usepackage[a4paper,twoside,margin=0.8in]{geometry}

\newcommand{\1}{\mathds{1}}
\newcommand{\ee}{\mathrm{e}}
\newcommand{\I}{\mathrm{i}}
\newcommand{\pluseq}{\mathrel{+}=} 
\newcommand{\minuseq}{\mathrel{-}=} 
\newcommand{\floor}[1]{\lfloor #1 \rfloor}
\newcommand{\ceil}[1]{\lceil #1 \rceil}
\newcommand{\dd}{\text{d}} 
\newcommand{\st}{\:\: | \:\:} 
\renewcommand{\epsilon}{\varepsilon}

\newcommand{\range}[1]{\text{Range}(#1)}

\newcommand{\Ac}{\mathcal{A}}
\newcommand{\Cc}{\mathcal{C}}

\newcommand{\Cb}{\mathbb{C}}

\newcommand{\Kb}{\mathbb{K}} 
\newcommand{\Nb}{\mathbb{N}}
\newcommand{\Rb}{\mathbb{R}}
\newcommand{\Zb}{\mathbb{Z}}
\newcommand{\Tb}{\mathbb{T}} 

\newcommand{\Be}[1]{\mathrm{Be}\left(#1\right)} 
\newcommand{\CP}{\mathrm{CPo}} 
\newcommand{\Po}[1]{\mathrm{Po}\left(#1\right)} 
\newcommand{\Normal}[2]{\mathcal{N}\left(#1,#2\right)}
\newcommand{\E}[1]{\mathbb{E}\left[#1\right]}
\renewcommand{\P}[1]{\mathbb{P}\left(#1\right)}

\newcommand{\var}[2]{\mathrm{VaR}_{#1}\left(#2\right)}
\newcommand{\es}[2]{\mathrm{ES}_{#1}\left(#2\right)}

\newtheorem{definition}{Definition}[section]
\newtheorem{lemma}[definition]{Lemma}
\newtheorem{theorem}[definition]{Theorem}
\newtheorem{proposition}[definition]{Proposition}
\theoremstyle{remark}
\newtheorem{remark}[definition]{Remark}
\newtheorem{example}[definition]{Example}

\newcommand{\id}{\mathrm{id}}
\newcommand{\lle}{\left[\!\left[} 
\newcommand{\rre}{\right]\!\right]}

\setlist[enumerate]{itemsep=10pt,topsep=10pt}
\setlist[itemize]{itemsep=5pt,topsep=5pt}

\begin{document}

\title[Mod-Poisson approximation schemes: applications to credit risk]{Mod-Poisson approximation schemes:\\ applications to credit risk}


\author{Pierre-Lo\"ic M\'eliot}
\address{Institut de math\'ematiques d'Orsay, Universit\'e Paris-Saclay, France}
\email{pierre-loic.meliot@universite-paris-saclay.fr}
\thanks{}

\author{Ashkan Nikeghbali}
\address{Institute of Mathematics, Universit\"at Z\"urich, Switzerland}
\email{ashkan.nikeghbali@math.uzh.ch}
\thanks{}

\author{Gabriele Visentin}
\address{Department of Mathematics, RiskLab, ETH Zürich, Switzerland}
\email{gabriele.visentin@math.ethz.ch}
\thanks{}


\date{\today}

\keywords{Mod-$\phi$ convergence, Mod-$\phi$ approximation schemes, Credit risk, Risk measures, CDO pricing, Poisson approximation}

\begin{abstract}
We introduce a new numerical approximation method for functionals of factor credit portfolio models based on the theory of mod-$\phi$ convergence and mod-$\phi$ approximation schemes. The method can be understood as providing correction terms to the classic Poisson approximation, where higher order corrections lead to asymptotically better approximations as the number of obligors increases. We test the model empirically on two tasks: the estimation of risk measures ($\mathrm{VaR}$ and ES) and the computation of CDO tranche prices. We compare it to other commonly used methods -- such as the recursive method, the large deviations approximation, the Chen--Stein method and the Monte Carlo simulation technique (with and without importance sampling) -- and we show that it leads to more accurate estimates while requiring less computational time.
\end{abstract}

\maketitle

\tableofcontents

\section{Introduction}
\label{sec:introduction}

Mod-$\phi$ convergence \cite{feray2016mod} is a new notion of convergence for sequences of random variables which provides a unified framework for the derivation of refinements of classical limit theorems, such as the central limit theorem, the Berry--Esseen theorem, precise large and moderate deviation results, local limit theorems and more.
The results presented in this paper rely on mod-$\phi$ approximation schemes, which were first introduced in \cite{barbour2014mod} and \cite{chaibi2020mod}, together with many examples of applications to probability theory, analytic number theory and combinatorics. 
\medskip 

The paper is organized as follows. Sections \ref{sec:mod-phi_convergence} and \ref{sec:mod-phi_approximation_schemes} provide a self-contained introduction to the main concepts of mod-$\phi$ convergence and mod-$\phi$ approximation schemes respectively. Section \ref{sec:mod-poisson_approximation} discusses the application of mod-Poisson approximation schemes to credit portfolio models. The main result is Theorem \ref{thm:modPoisson_convergence}, which states that the total number of portfolio defaults $L_n = \sum_{i=1}^n Y_i$ converges mod-Poisson conditionally on the mixing factor and can therefore be approximated using mod-Poisson approximation schemes $(\nu^{(r)}_n)_{n \in \Nb}$, where $r \ge 1$ is the order of approximation. Higher orders lead to better asymptotic approximations and can therefore be used to improve accuracy for finite $n$. The performance of this approximation is empirically tested on two benchmark applications -- the estimation of risk measures in Section \ref{sec:risk_measures_estimation} and the pricing of synthetic CDO tranches in Section \ref{sec:cdo_pricing} -- and is compared in terms of accuracy and computational time to the following commonly used estimation methods: recursive methodology \cite{hull2004valuation, brasch2004note}, the large deviations approximation \cite{dembo2004large}, Chen--Stein's method and the zero-bias transformation method \cite{el2008gauss, el2009stein} and Monte Carlo simulation, with and without importance sampling \cite{glasserman2005importance}. These methods and the corresponding algorithms are presented in Appendix \ref{sec:estimation_overview}. Section \ref{sec:mod-compound_poisson_approximation} presents the derivation of mod-compound Poisson approximation schemes for the case of portfolio losses $L_n = \sum_{i=1}^n Z_i\, Y_i$ with conditionally i.i.d.~exposures $(Z_i)_{i=1}^n$. Section \ref{sec:conclusions} concludes, and Appendices \ref{app:useful_formulae}, \ref{sec:mobius_function} and \ref{sec:stirling} present the detailed proofs of certain theoretical results. In particular, the framework of mod-$\phi$ convergence relies on certain combinatorial arguments related to the theory of Möbius inversion and to the theory of symmetric functions; for this later topic, we refer to \cite[Chapter 1]{Mac95}, and everything required is recalled in Section \ref{sec:stirling}.

\section{Mod-\texorpdfstring{$\phi$}{phi} convergence}
\label{sec:mod-phi_convergence}

Given a $\Zb$-valued random variable $X$ with probability law $\mu_X$, we can define its characteristic function as the Fourier transform of its law:
$$ \widehat{\mu}_X(\xi) = \E{\ee^{\I\xi X}} = \sum_{k \in \Zb} \mu_X(\{k\})\, \ee^{\I k \xi}, \quad \xi \in \Tb := \Rb/2 \pi \Zb.$$
Notice that the Fourier transform $\widehat{\mu}_X$ in this case is well-defined on $\Tb$, because we are assuming $X$ to be an integer-valued random variable.
\medskip 

If $\mu$ is an infinitely divisible distribution on $\Zb$, then its Fourier transform admits the following representation: 
$$\widehat{\mu}(\xi) = \ee^{\phi(\xi)}, \quad \xi \in \Tb,$$
where $\phi$ is a periodic function of period $2 \pi$, called the L\'{e}vy--Khintchine exponent \cite{steutel2003infinite} of the distribution. Infinitely divisible laws will play a fundamental role in mod-$\phi$ convergence. Two very important examples are given below.

\begin{example}{(\textbf{Poisson distribution})} Let $X$ be a Poisson random variable with parameter $\lambda$, i.e.~$X \sim \Po{\lambda}$. Then,  its L\'evy--Khintchine exponent is given by:
$$ \phi(\xi) = \lambda (\ee^{\I\xi} - 1), \quad \xi \in \Tb.$$
\end{example}

\begin{example}{(\textbf{Compound Poisson distribution})}
\label{ex:compound_poisson}
A random variable $X$ follows a compound Poisson distribution if it admits the following representation:
$$ X = \sum_{i=1}^N Z_i,$$
where $(Z_i)_{i=1}^\infty$ is a family of i.i.d.~$\Nb$-valued random variables distributed like $Z$, and $N$ is an independent Poisson random variable with parameter $\lambda$. In this case we write $X \sim \CP(\lambda, Z)$.
If $X \sim \CP(\lambda, Z)$, then its L\'evy--Khintchine exponent is given by:
$$
    \phi(\xi) := \log\left( \E{\ee^{\I\xi X}} \right) = \lambda(\widehat{\mu}_Z(\xi) - 1),
$$
where $\widehat{\mu}_Z$ is the Fourier transform of the law of $Z$.
For $\Nb$-valued random variables, it turns out that every infinitely divisible distribution is a compound Poisson distribution: if $X$ is a $\Nb$-valued infinitely divisible distribution, then it necessarily follows that there exists a $\lambda > 0$ and a $\Nb$-valued random variable $Z$ such that $X \sim \CP(\lambda, Z)$; see \cite[Theorem 3.2]{steutel2003infinite}.
\end{example}

By L{\'e}vy's continuity theorem \cite[Thm. 4.3]{kallenberg1997foundations}, knowledge of $\widehat{\mu}_X$ is equivalent to knowledge of the full probability law $\mu_X$ and limit theorems for probability laws can be  derived directly in Fourier space in terms of pointwise convergence of characteristic functions. For instance, the so-called \emph{law of small numbers} (which is a generalization of the original Poisson convergence theorem) states that if $Y_{i,n}, i = 1, \ldots, n$ is a triangular array of independent Bernoulli random variables with success probabilities $p_{i,n}$ such that $\sum_{i=1}^n p_{i,n} \to \lambda \in \Rb$ and $\max_{1 \le i \le n} p_{i,n} \to 0$, then $X_n := \sum_{i=1}^n Y_{i, n}$ converges in law to a random variable $Y$, with $Y \sim \Po{\lambda}$.
\medskip

This convergence can be proved from the pointwise convergence of $\widehat{\mu}_{X_n}$ to the characteristic function of $Y$ as $n \to \infty$ (see, for instance, \cite[Theorem 3.6.1]{durrett2010probability}).
But this limit law can also be interpreted as a non-asymptotic approximation result, stating that the law of $X_n$ can be approximated by a Poisson variable with parameter $\lambda_n = \sum_{i=1}^n p_{i,n}$, which is just a sum of $\lambda_n$ independent Poisson random variables with parameter one. We would expect this approximation to work well for $n \to \infty$, or equivalently for $\lambda_n \to \infty$.
\medskip 

More generally, given a sequence $(X_n)_{n \in \Nb}$ of $\Zb$-valued random variables, we are interested in the problem of approximating the law $\mu_{X_n}$ of $X_n$ -- which might be difficult to compute or simulate -- by a sum $Y_n$ of $\lambda_n$ i.i.d.~copies of a given $\Zb$-valued infinitely divisible law with L\'{e}vy--Khintchine exponent $\phi$, for $\lambda_n \to \infty$. Let us notice that we can make sense of the law of $Y_n$ even if $\lambda_n$ is not an integer, by placing oneself in the convolution semigroup of probability measures generated by $\phi$. This approximation problem leads to the following definition:

\begin{definition}[Mod-$\phi$ convergence]
\label{def:mod-phi_convergence}

We say that $(X_n)_{n \in \Nb}$ converges mod-$\phi$ with parameters $(\lambda_n)_{n \in \Nb}$ and limiting function $\psi$ if $\lambda_n \to \infty$ and
\begin{equation}
\label{eq:deconv_residues}
\widehat{\mu}_{X_n}(\xi)\, \ee^{-\lambda_n \phi(\xi)} = \psi_n(\xi), \quad \xi \in \Tb    
\end{equation}
with 
$$ \lim_{n \to \infty} \psi_n(\xi) = \psi(\xi).$$
\end{definition}
The convergence $\psi_n \to \psi$ usually occurs in a space of continuously derivable functions $\mathcal{C}^r(\Tb)$ endowed with the norm $\| f \|_{\Cc^r} = \sup_{|\alpha| \le r} \sup_{\xi \in \Tb} |\partial^\alpha f (\xi)|$. When a sequence $(X_n)_{n \in \Nb}$ converges mod-$\phi$, we can use the functions $\psi_n$ and $\psi$ to extract precise information about the behavior of $\mu_{X_n}$, both asymptotically and non-asymptotically, and to construct good approximations by the reference infinitely divisible law, as we will see in the next section.

\section{Mod-\texorpdfstring{$\phi$}{phi} approximation schemes}
\label{sec:mod-phi_approximation_schemes}

The right space to study these approximations turns out to be the Wiener algebra $\Ac(\Tb)$ of absolutely convergent Fourier series on $\Tb$, which is a Banach algebra under the pointwise product and the norm
$$\|\widehat{\mu}\|_{\Ac(\Tb)} = \sum_{n \in \Zb} |c_n(\widehat{\mu})|,$$
where $c_n(\widehat{\mu}) := \int_\Tb \widehat{\mu}(\xi)\, \ee^{-\I n\xi} \,d\xi$ is the $n$-th Fourier coefficient of $\widehat{\mu}$.
\medskip 

Working in the Wiener algebra $\Ac(\Tb)$ is essential for at least two reasons.
First of all, the characteristic function $\widehat{\mu}$ of a probability law $\mu$ always lies in $\Ac(\Tb)$ and its Fourier coefficients satisfy $c_n(\widehat{\mu}) = \mu(n)$, so that the Wiener algebra norm of $\widehat{\mu}$ is actually equal to the total variation norm of the law $\mu$ itself (or twice this norm depending on the chosen convention). This means that good approximation bounds in the Wiener algebra directly translate into good approximation bounds in total variation distance.\medskip 

Second, Wiener's $1/f$ theorem guarantees that $\ee^{-\lambda_n \phi(\xi)}$ lies in $\Ac(\Tb)$. This fact, together with Eq.~\eqref{eq:deconv_residues}, implies that $\psi_n \in \Ac(\Tb)$ and that it can be thought of as the deconvolution residue of the law $\mu_{X_n}$ by a sum of $\lambda_n$ independent copies of the infinitely divisible reference law with L\'evy--Khintchine exponent $\phi$. In other words, we can think of $\psi_n$ as that element of $\Ac(\Tb)$ that satisfies the following equation:
\begin{equation}
\label{eq:mod-phi}
\widehat{\mu}_{X_n}(\xi) = \psi_n(\xi)\, \ee^{\lambda_n \phi(\xi)}.    
\end{equation}
Clearly Equation \eqref{eq:mod-phi} implies that we can always use the deconvolution residue $\psi_n$ to reconstruct $\widehat{\mu}_{X_n}$ perfectly by pointwise multiplication with the Fourier transform of the reference infinitely divisible, but in practice $\psi_n$ might be as hard to compute as $\widehat{\mu}_{X_n}$ itself. Nevertheless one can construct good approximations for $\widehat{\mu}_{X_n}$ by substituting $\psi_n$ in Eq.~\eqref{eq:mod-phi} with another function $\chi_n \in \Ac(\Tb)$ that approximates it sufficiently well on $\Tb$ and that is easier to compute. This leads to the following definition of mod-$\phi$ approximation schemes, which was first given in \cite{chaibi2020mod}. In the sequel, we focus on $\Nb$-valued random variables, hence non-negative. This choice fits with the application that we have in mind, namely, the approximation of the distribution of the total loss variable of a credit portfolio. However, notice that if we were working with variables that can be positive or negative, then a straightforward extension of our methods to $\Zb$-valued random variables exist and is described in \cite[Example 1.11]{chaibi2020mod}.

\begin{definition}[Mod-$\phi$ approximation scheme of order $r$]
Let $(X_n)_{n \in \Nb}$ be a sequence of $\Nb$-valued random variables that converges mod-$\phi$ with parameters $(\lambda_n)_{n \in \Nb}$. We suppose that $\psi_n$ admits the following series expansion around zero:
$$
\psi_n(\xi) = 1 + \sum_{k=1}^\infty b_{k, n} (\ee^{\I \xi} - 1)^k.    
$$
Then the mod-$\phi$ approximation scheme of order $r$ for $(X_n)_{n \in \Nb}$ is a sequence of discrete signed measures $(\nu^{(r)}_n)_{n \in \Nb}$ on $\Zb$, such that
\begin{equation}
\label{eq:reconvolution_of_residues}
\widehat{\nu}^{(r)}_n(\xi) = \chi^{(r)}_n(\xi) \,\ee^{\lambda_n \phi(\xi)},
\end{equation}
where $\chi^{(r)}_n$ is the polynomial of degree $r$ that approximates $\psi_n$ around $0$ up to order $r$:
$$\chi^{(r)}_n (\xi) = 1 + \sum_{k=1}^r b_{k, n} (\ee^{\I\xi} - 1)^k .$$
\end{definition}

An explicit formula for the measure $\widehat{\nu}_n^{(r)}$ in terms of $\lambda_n$ and of the coefficients $b_{k,n}$ is given in \cite[Lemma 3.8 and Remark 3.9]{chaibi2020mod}. For instance, if $\phi=\Po{1}$ and $r=2$, then we have:
$$\nu^{(2)}_n(k) = \frac{\ee^{-\lambda_n}\,(\lambda_n)^k}{k!}\,\left(1 + b_{2,n}\,\left(1-\frac{2k}{\lambda_n} + \frac{k(k-1)}{(\lambda_n)^2}\right)\right),$$
so $\nu^{(2)}_n$ is in this case a perturbation of $\nu^{(0)}_n = \Po{\lambda_n}$. In general, 
the zero-th order approximation $\nu_n^{(0)}$ corresponds to approximating $\mu_{X_n}$ with a sum of $\lambda_n$ i.i.d.~copies of the reference infinitely divisible law, while higher order approximations will correspond to signed measures that approximate the law $\mu_{X_n}$ increasingly well. Let us be a bit more precise on this claim. In a companion paper \cite{companion}, we prove  that if $X_n$ is a sum of $n$ independent Bernoulli variables with parameters $p_1,\ldots,p_n$, then the corresponding approximations $\nu_n^{(r)}$ satisfy
\begin{equation} 
\sum_{k \in \Nb} |\P{X_n=k} - \nu_n^{(r)}(k)| \leq A \left(\frac{K\sigma_n}{\sqrt{\lambda_n}}\right)^{r+1},\label{eq:companion}
\end{equation}
with $\lambda_n = \sum_{i=1}^n p_i$, $(\sigma_n)^2 = \sum_{i=1}^n (p_i)^2$ and $A$ and $K$ universal constants. The mod-Poisson convergence of such sequences $(X_n)_{n \in \Nb}$ is explained in Theorem \ref{thm:modPoisson_convergence} below. Thus,  the quality of the approximation $\nu_n^{(r)}$ of $\mu_{X_n}$ indeed increases with the order of approximation $r$.  The results that we shall present in Sections \ref{sec:risk_measures_estimation} and \ref{sec:cdo_pricing} are numerical evidences of the general theoretical result \eqref{eq:companion}, and applications for the study of credit risk models.\medskip
 
In general, the way in which the measures $\nu^{(r)}_n$ incorporate the information contained in $\mu_{X_n}$ and in the residue $\psi_n$ as $r$ increases can also be understood in terms of factorial cumulants.

\begin{definition}[Factorial cumulant generating function]
If $X$ is a $\Nb$-valued random variable, then its factorial cumulant generating function (when it exists, which is always the case in this paper) is defined as
$$ \log\left( \E{(z+1)^X} \right) = \sum_{k=1}^\infty \frac{1}{k!}\, \kappa_k(X)\, z^k, \quad z \in \Cb,$$
and the coefficient $\kappa_k(X)$ is called the $k$-th factorial cumulant of $X$.
\end{definition}


\begin{example}
If $X \sim \Po{\lambda}$, then 
$$ \kappa_k(X) = \begin{cases} \lambda & \text{if $k=1$,} \\ 0 & \text{otherwise.} \end{cases} $$
\end{example}

We note that if $X$ is a random variable with law $\mu$, then the coefficients in the expansion of its Fourier transform $\widehat{\mu}(\xi)$ in powers of $z := (\ee^{\I\xi} - 1)$ are precisely its factorial cumulants, so that the coefficients of a mod-$\phi$ approximation scheme are naturally related to these quantities.
In particular, the following proposition shows that the measures $\nu^{(r)}_n$ achieve better approximations by matching exactly the factorial cumulants of $X_n$ up to order $r$, while maintaining the factorial cumulants of the reference infinitely divisible law for all higher orders. 

\begin{proposition}
\label{prop:factorial_approximation}
Let $(X_n)_{n \in \Nb}$ be a sequence of $\Nb$--valued random variables that converges mod-$\phi$ with parameters $(\lambda_n)_{n \in \Nb}$, and let $(Y_n)_{n \in \Nb}$ follow the reference infinitely divisible laws with exponents $(\lambda_n \phi)_{n \in \Nb}$.
If $(\nu_n^{(r)})_{n \in \Nb}$ is a mod-$\phi$ approximation scheme of order $r$ for $(X_n)_{n \in \Nb}$ and if we denote by $\kappa^{(r)}_{k,n}$ the $k$-th factorial cumulant of $\nu^{(r)}_n$, then 
$$ \kappa^{(r)}_{k,n} = \begin{cases} \kappa_k(X_n) & \text{for $k = 1, \ldots, r$,} \\ \kappa_k(Y_n) & \text{for $k \ge r+1$,} \end{cases}$$
where $\kappa_k(X)$ is the $k$-th factorial cumulant of the random variable $X$.
\end{proposition}

\begin{proof}
From Equation \eqref{eq:mod-phi} we compute:
\begin{align*}
\psi_n(\xi) & = \exp \left( \sum_{k=1}^\infty \frac{1}{k!} \left( \kappa_k(X_n) - \kappa_k(Y_n) \right) z^k \right) \\
& = 1 + \sum_{k=1}^\infty \frac{1}{k!} \sum_{\pi \in \Pi(k)} \prod_{B \in \pi} \left( \kappa_{|B|}(X_n) - \kappa_{|B|}(Y_n) \right) z^k,
\end{align*}
by using the first identity of Theorem \ref{thm:exp_log} in the last step. Since $\chi^{(r)}_n$ approximates $\psi_n$ up to order $r$ in powers of $z$, we have:
$$ \log \left( \chi^{(r)}_n(\xi) \right) = \log \left( 1 + \sum_{k=1}^r \frac{1}{k!}\sum_{\pi \in \Pi(k)} \prod_{B \in \pi} \left( \kappa_{|B|}(X_n) - \kappa_{|B|}(Y_n) \right) z^k \right).$$
We can now compute the coefficients of the series expansion of $\log(\chi^{(r)}_n)$ in powers of $z$ by using the second identity of Theorem \ref{thm:exp_log}:
\begin{align*}
    s!\,[z^s] \log \left( \chi^{(r)}_n(\xi) \right) & = \sum_{\sigma \in \Pi(s)} \mu(\sigma, \widehat{1}_s) \prod_{D \in \sigma} \sum_{\pi \in \Pi(|D|)} \prod_{B \in \pi} \left( \kappa_{|B|}(X_n) - \kappa_{|B|}(Y_n) \right) \1_{\{|B| \le r\}} \\
    & = \sum_{\sigma \in \Pi(s)} \mu(\sigma, \widehat{1}_s) \sum_{\tau \le \sigma} \prod_{B \in \tau} \left( \kappa_{|B|}(X_n) - \kappa_{|B|}(Y_n) \right) \1_{\{|B| \le r\}} \\
    & = \sum_{\tau \in \Pi(s)} \sum_{\sigma \in \Pi(s)} \zeta(\tau, \sigma)\, \mu(\sigma, \widehat{1}_s) \prod_{B \in \tau} \left( \kappa_{|B|}(X_n) - \kappa_{|B|}(Y_n) \right) \1_{\{|B| \le r\}}\\
    & = \sum_{\tau \in \Pi(s)} \delta(\tau, \widehat{1}_s)  \prod_{B \in \tau} \left( \kappa_{|B|}(X_n) - \kappa_{|B|}(Y_n) \right) \1_{\{|B| \le r\}} \\
    & = \left( \kappa_{s}(X_n) - \kappa_{s}(Y_n) \right) \1_{\{s \le r\}}
\end{align*}
where we have used the convolution relation $\zeta \star \mu= \delta$ for the M\"obius function of the poset $\Pi(s)$. We refer the reader to Appendix \ref{sec:mobius_function} for a primer on the M\"obius function and its basic properties.
\medskip 

Finally, we can compute the factorial cumulants of the measures $\nu^{(r)}_n$ by taking logarithms in Equation \eqref{eq:reconvolution_of_residues}:
\begin{align*}
    \log \left( \widehat{\nu}^{(r)}_n(\xi) \right) & = \sum_{k=1}^\infty \frac{1}{k!} \left( \kappa_k(Y_n) \right) z^k + \log \left(\chi^{(r)}_n(\xi) \right) \\
    & = \sum_{k=1}^\infty \frac{1}{k!} \left( \kappa_k(Y_n) \right) z^k + \sum_{k=1}^r \frac{1}{k!} \left( \kappa_{k}(X_n) - \kappa_{k}(Y_n) \right) z^k \\
    & = \sum_{k=1}^r \frac{1}{k!} \left( \kappa_k(X_n) \right) z^k + \sum_{k=r+1}^\infty \frac{1}{k!} \left( \kappa_k(Y_n) \right) z^k.
\end{align*}
\end{proof}
As a particular case of Proposition \ref{prop:factorial_approximation}, we remark that in the case of mod-Poisson approximation schemes, the signed measures $\nu^{(r)}_n$ have all factorial cumulants equal to zero (because the Poisson distribution itself does) with the exception of the first $r$ cumulants, which exactly match the factorial cumulants of $X_n$.
Since factorial cumulant generating functions fully characterize probability distributions, Proposition  \ref{prop:factorial_approximation} also implies that $\nu^{(\infty)}_n = \mu_{X_n}$.
\medskip

Despite the availability of this interpretation in terms of factorial cumulants, the signed measures $\nu^{(r)}_n$ are in general difficult to compute, even if we have full knowledge of the coefficients $(b_{k,n})_{k=1}^r$. Nevertheless, the computation of expectations of functions of these measures can be done efficiently, as the following proposition shows.

\begin{proposition}
\label{prop:expectation_approximation}
Let $(X_n)_{n \in \Nb}$ be a sequence of $\Nb$-valued random variables that converges mod-$\phi$ with parameters $(\lambda_n)_{n \in \Nb}$, and let $(\nu^{(r)}_n)_{n \in \Nb}$ be its mod-$\phi$ approximation scheme of order $r$. Then, for any bounded function $f: \Nb \to \Rb$, the integral of $f$ with respect to $\nu^{(r)}_n$ is given by:
$$
    \nu^{(r)}_n(f) = \sum_{j \in \Nb} f(j) \: \nu^{(r)}_n(\{j\}) = \E{f(Y_n)} + \E{\Delta_n(r,f)(Y_n)}
$$
where $Y_n$ follows the reference infinitely divisible law with exponent $\lambda_n \phi$, and where the correction term $\Delta_n(r,f)$ is given by:
$$
\Delta_n(r,f)(j) = \sum_{k=1}^r b_{k,n} (\Delta_+^k (f))(j).
$$
Here, $\Delta_+^k$ denotes the $k$-th power of the forward finite difference operator:
$$(\Delta_+^k (f))(j) = \sum_{l=0}^k (-1)^{k-l} \binom{k}{l} f(j + l).$$
\end{proposition}
\begin{proof}
A similar result with functions $f \in \ell^2(\Nb)$ is stated in \cite[Proposition 1.12]{chaibi2020mod}; in the sequel, we give a new proof when $f$ is only assumed to be bounded, and at the end we shall even explain how to extend the result to polynomially bounded functions. Since $\nu^{(r)}_n$ is a finite signed measure, it is in $\ell^1(\Nb) \subset \ell^2(\Nb)$. On the other hand, for any $f,g\in \ell^2(\Nb)$, the Parseval formula holds:
$$\sum_{j \in \Nb} f(j)\,g(j) = \int_{0}^{2\pi} \overline{\widehat{f}(\xi)}\,\widehat{g}(\xi)\,\frac{d\xi}{2\pi} .$$
In particular, if $f \in \ell^2(\Nb)$ and $g=\nu^{(r)}_n$, we obtain:
$$\nu^{(r)}_n(f)=\nu^{(0)}_n(f)+\sum_{k=1}^r b_{k,n}\left(\int_{0}^{2\pi} \overline{\widehat{f}(\xi)\,(\ee^{-\I\xi}-1)^k}\,\,\widehat{\nu}^{(0)}_n(\xi)\,\frac{d\xi}{2\pi}\right),$$
where $\widehat{\nu}^{(0)}_n(\xi)=\ee^{\lambda_n\phi(\xi)}$. However, for $f \in \ell^1(\Nb)$, we have 
\begin{align*}
    \widehat{f}(\xi)\,(\ee^{-\I \xi}-1) &= \sum_{j\in \Nb}f(j)\,\ee^{\I j \xi} (\ee^{-\I \xi}-1) \\
    &= f(0)\, \ee^{-\I \xi} +  \sum_{j\in \Nb}(f(j+1)-f(j))\,\ee^{\I j \xi}  \\ 
    &= f(0)\, \ee^{-\I \xi}+ \widehat{\Delta_+(f)}(\xi)
    \end{align*}
so
\begin{align*}
    \int_0^{2\pi} \overline{\widehat{f}(\xi)\,(\ee^{-\I \xi}-1)}\,\widehat{\nu}^{(0)}_n(\xi)\,\frac{d\xi}{2\pi} &= f(0) \int_{0}^{2\pi}\ee^{\I \xi}\,\widehat{\nu}^{(0)}_n(\xi)\,\frac{d\xi}{2\pi} + \int_0^{2\pi} \overline{\widehat{\Delta_+(f)}(\xi)}\,\,\widehat{\nu}^{(0)}_n(\xi)\,\frac{d\xi}{2\pi} \\ 
    &= \nu^{(0)}_n(\Delta_+(f))
\end{align*}
since the first integral vanishes ($\widehat{\nu}^{(0)}_n$ has only positive Fourier coefficients). By an immediate induction,
$$\int_{0}^{2\pi} \overline{\widehat{f}(\xi)\,(\ee^{-\I\xi}-1)^k}\,\,\widehat{\nu}^{(0)}_n(\xi)\,\frac{d\xi}{2\pi} = \nu^{(0)}_n(\Delta_+^k(f))$$
for any $k \geq 1$, whence the result for $f \in \ell^1(\Nb)$. So, we have the equality of linear forms on the space of summable real functions on $\Nb$:
$$\nu^{(r)}_n = \nu^{(0)}_n \circ \left(\id + \sum_{k=1}^r b_{k,n}\,\Delta_+^k\right).$$
Consider now a bounded function $f : \Nb \to \Rb$, and for $L \in \Nb$, denote $f_L(j) = \1_{j \leq L}\,f(j)$, which is in $\ell^1(\Nb)$. Since $\nu^{(0)}_n$ and $\nu^{(r)}_n$ are in $\ell^1(\Nb)$, for any sequence of functions $(g_L)_{L \in \Nb}$ which are uniformly bounded by a constant $K$ and such that $g_L(j) \to_{L \to \infty} g(j)$ for any $j \in \Nb$, $\nu^{(0)}_n(g_L) \to_{L \to \infty} \nu^{(0)}_n (g)$ and 
 $\nu^{(r)}_n(g_L) \to_{L \to \infty} \nu^{(r)}_n (g)$ by the dominated convergence theorem. Here, we have obviously 
 \begin{align*}
 &f_L(j) \to_{L \to \infty} f(j); \\
 &\left(f_L(j) + \sum_{k=1}^r b_{k,n}\,(\Delta_+^kf_L)(j)\right) \to_{L \to \infty} \left(f(j) + \sum_{k=1}^r b_{k,n}\,(\Delta_+^kf)(j)\right),
 \end{align*}
 so
 $$\nu^{(r)}_n(f) = \lim_{L \to \infty} \nu^{(r)}_n(f_L) =  \lim_{L \to \infty} \nu^{(0)}_n\left(f_L + \sum_{k=1}^r b_{k,n}\,\Delta_+^k(f_L)\right) = \nu^{(0)}_n\left(f + \sum_{k=1}^r b_{k,n}\,\Delta_+^k(f)\right).$$
 Let us remark that the dominated convergence argument works for a larger class of functions: a sufficient assumption is that $(f(j))_{j \in \Nb}$ and its shifts $(f(j+l))_{j \in \Nb}$ with $1\leq l\leq r$ are bounded by functions which are integrable against the reference infinitely divisible distribution $\nu^{(0)}_n$. In particular, if $\phi$ is the exponent of the Poisson distribution, then the formula of the proposition holds for any $f$ bounded by a polynomial function. More generally, if $\phi=\CP(\lambda,Z)$ with $\E{Z^r}<+\infty$, then the formula of the proposition holds for any $f$ bounded by a polynomial function with degree $r$. 
\end{proof}

We can summarize Proposition \ref{prop:expectation_approximation} by saying that expectations of the form $\E{f(X_n)}$ can be approximated with a mod-$\phi$ approximation scheme $(\nu^{(r)}_n)_{n \in \Nb}$ in two steps:
\begin{enumerate}
    \item Replace $X_n$ by the infinitely divisible random variable $Y_n$.
    \item Correct the function $f$ by adding the correction term $\Delta_n(r,f)$.
\end{enumerate}
The main advantage of this approximation procedure is that computing integrals of the reference infinitely divisible law is typically easier, as they may even admit closed-form expressions in terms of well-known special functions. 

\begin{remark}
If the reference infinitely divisible law $\phi$ has a moment of order $3$, and if the convergence of residues $\psi_n(\xi) \to \psi(\xi)$ occurs in the space $\mathcal{C}^{r+2}(\Tb)$, then \cite[Theorem 3.11]{chaibi2020mod} shows that for any $r \geq 0$ and any bounded function $f : \Nb \to \Rb$, the difference 
$|\E{f(X_n)} - \E{f(Y_n)} - \E{\Delta_n(r,f)(Y_n)}|$ goes to $0$ as $n$ goes to infinity, with a speed of convergence which improves with the order of approximation $r$. Indeed,
\begin{align*}
    |\E{f(X_n)} - \E{f(Y_n)} - \E{\Delta_n(r,f)(Y_n)}| &= |\mu_{X_n}(f) - \nu_n^{(r)}(f)| \\
    &\leq \|f\|_\infty\,d_{\mathrm{TV}}(\mu_{X_n},\nu_n^{(r)}) = O\left(\frac{1}{(\lambda_n)^{\frac{r+1}{2}}}\right).
\end{align*}
In the case which we shall examine in Section \ref{sec:mod-poisson_approximation} (sums of independent Bernoulli variables), the result follows also from the general estimate \eqref{eq:companion}.
\end{remark}

\section{Mod-Poisson approximation}
\label{sec:mod-poisson_approximation}

In the context of credit risk we are interested in approximating the total losses of a credit portfolio with $n$ counterparties, which is given by
$$
L_n = \sum_{i=1}^n Z_i \,Y_i,
$$
where $Y_i$ is the default indicator function for the $i$-th counterparty (i.e.~a Bernoulli random variable with $\P{Y_i = 1} = 1 - \P{Y_i = 0} = p_i$, where $p_i$ is the default probability) and $Z_i$ is the monetary loss incurred by the portfolio due to that counterparty's default.
\medskip

It is customary in credit risk to assume that the random variables $(Y_i)_{i=1}^n$ are conditionally independent given some underlying (macroeconomic or purely statistical) latent factor, $\Psi$, so that one can write $\P{Y_i = 1|\Psi} = p_i(\Psi)$, for some measurable function $p_i$, which depends on the particular credit risk model we are interested in studying.
On the other hand, the exposures $(Z_i)_{i=1}^n$ can be either constant or random. In the latter case, it is common practice to assume them to be i.i.d.~random variables, independent of $(Y_i)_{i=1}^n$. A less common choice is to incorporate the dependence on the underlying factor $\Psi$ by choosing $(Z_i)_{i=1}^n$ to be conditionally i.i.d.~given $\Psi$.
\medskip

In this section we focus on the case of constant unit exposures (i.e.~$Z_i = 1$, for all $i = 1, \ldots, n$), which is related to the classical Poisson approximation problem, but we anticipate that all results presented here will be extended in Section \ref{sec:mod-compound_poisson_approximation} to the general case of $(Z_i)_{i=1}^n$ conditionally i.i.d.~given $\Psi$.
The basic idea is to derive mod-Poisson convergence of $L_n$ conditionally on $\Psi$ and use mod-Poisson approximation schemes to estimate functionals of $L_n$ conditionally on $\Psi$. Unconditional estimates would then follow by integrating (numerically) on $\Psi$.

\begin{theorem}[\textbf{Mod-Poisson convergence}]
\label{thm:modPoisson_convergence}
Let $(L_n)_{n \in \Nb}$ be a sequence of total portfolio losses given by $L_n = \sum_{i=1}^n Y_i$, where we assume the $Y_i$'s to be independent random variables such that $Y_i \sim \Be{p_i}$.
\begin{enumerate}
    \item If $\sum_{i=1}^\infty p_i = +\infty$ and $\sum_{i=1}^\infty (p_i)^2 < +\infty$, then $(L_n)_{n \in \Nb}$ converges mod-$\phi$ with parameters $\lambda_n = \sum_{i=1}^n p_i$ and with reference infinitely divisible law the Poisson distribution $\Po{1}$ (i.e.~$\phi(\xi) = \ee^{\I\xi} - 1$).
    \item Furthermore, $(L_n)_{n \in \Nb}$ admits a mod-Poisson approximation scheme of order $r$ with the following coefficients:
\begin{equation}
\label{eq:mod-poisson_approximation_scheme_coefs}
 b_{k,n} = \frac{1}{k!} \sum_{\substack{\pi \in \Pi(k) \\ \forall B \in \pi,\,|B| \ge 2}} \mu(\widehat{0}_k, \pi) \left(\prod_{B \in \pi} \mathfrak{p}_{|B|, n}\right)
\end{equation}
where $\Pi(k)$ denotes the poset of set partitions of $\{1, \ldots, k\}$; $\mu(\cdot, \cdot)$ is the M\"obius function for the incidence algebra of the poset $\Pi(k)$; and $\mathfrak{p}_{k,n} := \sum_{j=1}^n (p_j)^k$ for $k \geq 2$.
\end{enumerate}
\end{theorem}

\begin{proof}
We compute the deconvolution residues as in Equation \eqref{eq:deconv_residues} using the explicit form for the characteristic function of $L_n$ and of the $\Po{1}$ distribution, obtaining:

\begin{align}
\psi_n(\xi) & = \widehat{\mu}_{L_n}(\xi) \,\ee^{-\sum_{j=1}^n p_j (\ee^{\I \xi} - 1)} \nonumber \\
& = \prod_{j=1}^n \left( 1 + p_j (\ee^{\I \xi} - 1) \right) \ee^{- p_j (\ee^{\I \xi}-1)} \nonumber \\
& = \exp \left( \sum_{j=1}^n \log \left( 1 + p_j (\ee^{\I \xi} - 1) \right) - p_j (\ee^{\I \xi} - 1) \right) \nonumber \\
& = \exp \left( \sum_{k=2}^\infty \frac{(-1)^{k-1}}{k} \left( \ee^{\I \xi} - 1 \right)^k \mathfrak{p}_{k,n} \right), \label{eq:psi_n}
\end{align}
where we have defined $\mathfrak{p}_{k,n} := \sum_{j=1}^n (p_j)^k$.
\medskip

The condition $\sum_{j=1}^\infty (p_j)^2 < +\infty$ guarantees that $\psi_n(\xi)$ converges to a limit $\psi(\xi)$ uniformly in $\Tb$, so the mod-$\phi$ convergence is proved.
In order to compute the coefficients $(b_{k,n})_{k=1}^r$ for the mod-$\phi$ approximation scheme of order $r$, we need to extract the coefficients of the series expansion of $\psi_n(\xi)$ in powers of $(\ee^{\I\xi}-1)$. This is again an application of the first identity of Theorem \ref{thm:exp_log}, since $\psi_n(\xi)$ is the exponential of the generating series $G(z) = \sum_{k=2}^\infty \frac{g_k}{k!}\,z_k$ with $g_{k\geq 2} = (-1)^{k-1} (k-1)!\,\mathfrak{p}_{k,n}$. Hence,
\begin{align*}
b_{k,n} = [z^k] \exp(G(z)) &= \frac{1}{k!} \sum_{\substack{\pi \in \Pi(k) \\\forall B \in \pi,\, |B| \ge 2}} \prod_{B \in \pi} g_{|B|} \\ 
&= \frac{1}{k!} \sum_{\substack{\pi \in \Pi(k) \\ \forall B \in \pi,\,|B| \ge 2}} \mu(\widehat{0}_k,\pi)\,\prod_{B \in \pi} p_{|B|,n}
\end{align*}
by using the formula for the Möbius function of the poset $\Pi(k)$ computed in the Appendix \ref{sec:mobius_function}.
Notice that the sum runs over all set partitions with each block of size at least two.
\end{proof}

\begin{remark}
\label{rk:partitions}
The computational time for the evaluation of the coefficients $(b_{k,n})_{k=1}^r$ can be substantially reduced by noticing that in \eqref{eq:mod-poisson_approximation_scheme_coefs} the term in the summation depends only on the type of the set partition $\pi$. This leads to the following equivalent, but computationally more advantageous expression:
$$b_{k,n} = \sum_{\substack{\lambda \in P(k) \\ \lambda=(\lambda_1\geq \lambda_2\geq\cdots \geq \lambda_\ell\geq 2)}} \frac{(-1)^{(k - \ell(\lambda))}}{z_\lambda}\,  \mathfrak{p}_{\lambda,n},$$
where $P(k)$ is the set of integer partitions of $k$, i.e.~the set of all finite non-increasing sequences of positive integers, $\lambda = (\lambda_1 \ge \lambda_2 \ge \cdots \ge \lambda_\ell)$, such that $\sum_{i=1}^\ell \lambda_i = k$. We also denote $\ell(\lambda)$ the number of parts of the integer partition $\lambda$; $\mathfrak{p}_{\lambda, n} = \prod_{i=1}^{\ell(\lambda)} \mathfrak{p}_{\lambda_i, n}$; and $z_\lambda = \prod_{k \ge 1} k^{m_k(\lambda)} (m_k(\lambda))!$, where $m_k(\lambda)$ is the number of parts of $\lambda$ of size $k$. In the expression of $b_{k,n}$ in terms of the $\mathfrak{p}_{\lambda,n}$, the sum runs over those integer partitions $\lambda \in P(k)$ such that all the parts of $\lambda$ are larger than $2$.
\end{remark}

The first few coefficients of the mod-$\phi$ approximation in terms of the obligors' default probabilities are:
\begin{align*}
    b_{1,n} & = 0 \qquad\qquad\,\,\,;\qquad 
    b_{2,n}  = - \frac{1}{2} \sum_{i=1}^n (p_i)^2; \\
    b_{3,n} & = \frac{1}{3} \sum_{i=1}^n (p_i)^3  \qquad;\qquad 
    b_{4,n} = - \frac{1}{4} \sum_{i=1}^n (p_i)^4 + \frac{1}{8} \left( \sum_{i=1}^n (p_i)^2 \right)^2.
\end{align*}

We remark that the first-order correction to the Poisson approximation presented in \cite{el2008gauss, el2009stein} and based on the Chen--Stein method corresponds to a mod-Poisson approximation scheme of order $r=2$. On the other hand, higher order approximations have so far remained inaccessible to the Chen--Stein method.

\begin{remark}
Equation \eqref{eq:mod-poisson_approximation_scheme_coefs} shows that the coefficients $b_{k,n}$ are symmetric functions in the default probabilities $p_i$, because they are polynomials in the Newton power sums $\mathfrak{p}_{k,n}$. By using the combinatorics of symmetric functions, one can rewrite the coefficients $(b_{k,n})_{k=1}^r$ as polynomials in the moments $(\E{(L_n)^k})_{k=1}^r$ of the total portfolio loss variable. Thus, if $M_{k,n}=\E{(L_n)^k}$ (which can be estimated numerically), then
\begin{equation}
b_{k,n} =  \frac{(-1)^{k}}{k!}\,(M_{1,n})^k+\sum_{1\leq m \leq l\leq k}  \frac{(-1)^{k-m}}{(k-l)!\,l!}\, \genfrac[]{0pt}{0}{l}{m}\,M_{m,n}(M_{1,n})^{k-l},\label{eq:bkn_in_terms_of_Mkn}
\end{equation}
where $\genfrac[]{0pt}{1}{l}{m}$ is the Stirling number which counts the permutations of size $l$ with exactly $m$ disjoint cycles (taking into account the fixed points as cycles with length $1$). In particular, knowledge of these first moments is sufficient in order to construct the $r$-th order approximation scheme $\nu_n^{(r)}$, and we do not need to know all the individuals default probabilities. The first coefficients are:
\begin{align*}
    b_{1,n} &=0  \qquad;\qquad b_{2,n}=\frac{1}{2}\,(M_{2,n}-M_{1,n}-(M_{1,n})^2); \\ 
    b_{3,n}&=\frac{1}{2}\,((M_{1,n})^2-M_{2,n}-M_{2,n}M_{1,n}) + \frac{1}{3}\,(M_{1,n} + (M_{1,n})^3 ) + \frac{1}{6}\,M_{3,n}.
\end{align*}
The proof of Eq.~\eqref{eq:bkn_in_terms_of_Mkn} is given in Appendix \ref{sec:stirling}.
\end{remark}
\medskip

The following proposition provides explicit closed-form expressions for the mod-Poisson approximations of the expectation of two common functions in credit risk, namely the tail function and the call function, which will be used in the numerical simulations of Sections \ref{sec:risk_measures_estimation} and \ref{sec:cdo_pricing}.

\begin{proposition}[Estimation formul{\ae}]\label{prop:estimation_formulae}

Let $\gamma$ be the lower incomplete gamma function given by:
$$ \gamma(x, \lambda) = \int_0^\lambda t^{x-1} \ee^{-t} dt.$$
Then:

\begin{enumerate}
    \item \textbf{Tail function}. Fix a value $x \in \Rb$, and let $f(j) = \1_{\{j > x\}} = \1_{\{j > \lfloor x \rfloor\}}$. Then,
\begin{equation}
\label{eq:tail_function_approximation}
\nu^{(r)}_n(f) = \frac{1}{\floor{x}!}\, \gamma \left( \floor{x} + 1, \lambda_n \right) + \ee^{- \lambda_n} \sum_{j = \floor{x} - r + 1}^{\floor{x}} \frac{(\lambda_n)^j}{j!}\, \Delta_n(r,f)(j).
\end{equation}
where $\Delta_n(r,f)$ is the correction term for $f$ as in Proposition \ref{prop:expectation_approximation}. 
\medskip   

\item \textbf{Call function}. Fix a value $K \in \Rb$, and let $f(j) = (j - K)^+$,  which is used in the pricing of call options and CDOs. Then,
\begin{align}
\label{eq:call_function_approximation}
\nu^{(r)}_n(f) = & \frac{\lambda_n}{(\ceil{K} -2)!}\, \gamma \left( \ceil{K} - 1, \lambda_n \right) - \frac{K}{(\ceil{K} - 1)!}\, \gamma \left( \ceil{K} , \lambda_n \right) \nonumber \\
& + \ee^{- \lambda_n} \sum_{j = \floor{K} - r + 1}^{\floor{K}} \frac{(\lambda_n)^j}{j!}\, \Delta_n(r,f)(j)
\end{align}
where $\Delta_n(r,f)$ is the correction term for $f$ as in Proposition \ref{prop:expectation_approximation}.
\end{enumerate}

\end{proposition}

\begin{proof}
We must compute in each case the following expectations:
\begin{equation}
\label{eq:mod_phi_approx_functions}
    \nu^{(r)}_n(f) = \E{f(Y_n)} + \E{\Delta_n(r,f)(Y_n)},
\end{equation}
where $Y_n \sim \Po{\sum_{i=1}^n p_i}$. Indeed, for the tail function,  Proposition \ref{prop:expectation_approximation} applies readily because $f$ is bounded. For the call function, we can use the remark at the end of the proof of Proposition \ref{prop:expectation_approximation}: the formula is valid because $f$ is bounded by a polynomial function.
The term $\E{f(Y_n)}$ can be computed in closed form for both functions. Indeed the tail function of a Poisson random variable is known to admit a closed-form expression in terms of the incomplete lower gamma function (see Proposition \ref{prop:poisson_tail_gamma} for the formula and its proof), while for the call function, one can proceed as follows:
\begin{align*}
    \E{(Y_n - K)^+} & = \sum_{j=0}^\infty (j - K)^+\, \ee^{-\lambda_n}\, \frac{\lambda_n^j}{j!} \\
    & = \sum_{j=\ceil{K}}^\infty (j - K)\, \ee^{-\lambda_n}\, \frac{(\lambda_n)^j}{j!} \\
    & = \ee^{-\lambda_n} \sum_{j=\ceil{K}}^\infty  \frac{\lambda_n^j}{(j-1)!} - K\, \P{Y_n > \ceil{K} - 1} \\
    & = \ee^{-\lambda_n} \lambda_n \sum_{j=\ceil{K}-1}^\infty \frac{(\lambda_n)^j}{j!} - K \,\P{Y_n > \ceil{K} - 1} \\
    & = \lambda_n \,\P{Y_n > \ceil{K} - 2} - K\, \P{Y_n > \ceil{K} - 1}. \\
\end{align*}

The term $\E{\Delta_n(r,f)(Y_n)}$ of Equation \eqref{eq:mod_phi_approx_functions} corresponds to the summations in Equations \eqref{eq:tail_function_approximation} and \eqref{eq:call_function_approximation}, the only difference being that the integration with respect to the distribution of $Y_n$ has been explicitly restricted to the integer-valued interval $\llbracket \floor{x} - r + 1, \floor{x} \rrbracket $. This follows from the fact that the correction term $\Delta_n(r,f)$ actually vanishes outside that interval for both the tail and the call function, as can be seen by their explicit formul{\ae}. Indeed, for the tail function one has
$$ \Delta_n(r,f)(j) = \sum_{k=1}^r b_{k,n} \sum_{l=0}^k (-1)^{k-l} \binom{k}{l} \1_{\{j + l > \floor{x}\}}.$$
On the one hand if $j \le \floor{x} - r$, then $j + l \le \floor{x}$ for all values of $l$, so $\Delta_n(r,f)(j)$ vanishes. On the other hand, if $j \ge \floor{x} + 1$, then $j + l > \floor{x}$ for all values of $l$, so that the inner summation yields $ \sum_{l=0}^k (-1)^{k-l} \binom{k}{l} = 0$ for all $k=1, \ldots, r$ and $\Delta_n(r,f)(j)$ is zero.
\medskip 

For the call function, instead, we have that
$$ \Delta_n(r,f)(j) = \sum_{k=1}^r b_{k,n} \sum_{l=0}^k (-1)^{k-l} \binom{k}{l} (j + l - K)^+.$$
On the one hand if $j \le \floor{K} - r$, then $j + l \le \floor{K}$ for all values of $l$, so $(j + l - K)^+$ is identically zero and $\Delta_n(r,f)(j)$ vanishes. On the other hand, if $j \ge \ceil{K}$, then $j + l \ge \ceil{K}$ for all values of $l$, so that the inner summation becomes $\sum_{l=0}^k (-1)^{k-l} \binom{k}{l} (j + l - K)$ which vanishes because $k \ge 2$ ($b_{1, n} = 0$) and $\sum_{l=0}^k (-1)^{k-l} \binom{k}{l} P(l) = 0$ for all polynomials $P$ of order strictly lower than $k$ (see, for instance, Corollary 2 in \cite{ruiz1996algebraic}).
\end{proof}

Formul{\ae} \eqref{eq:tail_function_approximation} and \eqref{eq:call_function_approximation} are particularly suitable for numerical implementation. Indeed, the lower incomplete gamma function can be estimated efficiently from its power series expansion (as shown in \cite{temme1994set}, on which the \texttt{python} implementation of the gamma function in \texttt{scipy} is based) while the correction term $\Delta_n(r,f)$, which requires numerical integration, is non-zero on at most $r$ points and is therefore easy to integrate.

\section{Mod-compound Poisson approximation}
\label{sec:mod-compound_poisson_approximation}

This section deals with the extension of the results presented in Section \ref{sec:mod-poisson_approximation} to the case of credit portfolios $L_n = \sum_{i=1}^n Z_i\, Y_i$ with random i.i.d.~exposures $(Z_i)_{i=1}^n$. In this case it is possible to prove a mod-compound Poisson convergence of the sequence $(L_n)_{n \in \Nb}$, as shown in the following theorem.

\begin{theorem}[\textbf{Mod-compound Poisson convergence}]
\label{thm:modCompoundPoisson_convergence}
Define a sequence $(L_n)_{n \in \Nb}$ of total portfolio losses given by $L_n = \sum_{i=1}^n Z_i \,Y_i$, where we assume the $Y_i$'s and $Z_i$'s to be mutually independent, $Y_i \sim \Be{p_i}$, and the $Z_i$'s are i.i.d.~random variables with the same law as $Z$, for $Z$ a given $\Nb$-valued random variable. 
\begin{enumerate}
    \item If $\sum_{i=1}^\infty p_i = +\infty$ and $\sum_{i=1}^\infty (p_i)^2 < +\infty$, then $(L_n)_{n \in \Nb}$ converges mod-$\phi$ with parameters $\lambda_n = \sum_{i=1}^n p_i$ and with reference infinitely divisible law the compound Poisson distribution $\CP(1, Z)$.
    \item Furthermore, $(L_n)_{n \in \Nb}$ admits a mod-compound Poisson approximation scheme of order $r$ with the following coefficients:

$$ b_{k,n} = \frac{1}{k!} \sum_{\sigma \in \Pi(k)} \sum_{\substack{\tau \in \Pi(k) \\ \tau \le \sigma}} \mu(\tau, \sigma) \prod_{D \in \tau} \E{(Z)_{|D|}} \prod_{B \in \sigma} \mathfrak{p}_{n^\sigma_\tau(B), n} ,$$
where $\Pi(k)$ is the poset of set partitions of $\{1, \ldots, k\}$; $n^\sigma_\tau(B)$ denotes the number of blocks of $\tau$ contained in the block $B$ of $\sigma$; $\mu(\tau, \sigma) = (-1)^{|\tau| - |\sigma|} \prod_{B \in \sigma} (n^\sigma_\tau(B) - 1)!$ is the M\"obius function for the incidence algebra of the poset $\Pi(k)$, and $\mathfrak{p}_{k,n} := \sum_{j=1}^n (p_j)^k$, with the convention $\mathfrak{p}_{1,n} := 0$.
\end{enumerate}
\end{theorem}

\begin{proof}
We compute the deconvolution residue as in Equation \eqref{eq:deconv_residues}, using the explicit form of the L\'evy--Khintchine exponent of $\CP(1,Z)$. We obtain:

\begin{align}
\psi_n(\xi) & = \widehat{\mu}_{L_n}(\xi)\, \ee^{-\sum_{i=1}^n p_i (\widehat{\mu}_Z(\xi) - 1)} \nonumber \\
& = \prod_{j=1}^n \left( 1 + p_j (\widehat{\mu}_Z(\xi) - 1) \right) \ee^{- p_j (\widehat{\mu}_Z(\xi)-1)} \nonumber \\
& = \exp \left( \sum_{j=1}^n \log \left( 1 + p_j (\widehat{\mu}_Z(\xi) - 1) \right) - p_j (\widehat{\mu}_Z(\xi) - 1) \right) \nonumber \\
& = \exp \left( \sum_{k=1}^\infty \frac{(-1)^{k-1}}{k} \left( \widehat{\mu}_Z(\xi) - 1 \right)^k \mathfrak{p}_{k,n} \right) \label{eq:residue_mod-compound_poisson},
\end{align}
where we have defined $\mathfrak{p}_{k,n} := \sum_{j=1}^n (p_j)^k$, with the convention $\mathfrak{p}_{1,n} := 0$.
\medskip

The condition $\sum_{j=1}^\infty (p_j)^2 < +\infty$ guarantees that $\psi_n(\xi)$ converges to a limit, $\psi(\xi)$, uniformly in $\Tb$, so the mod-$\phi$ convergence is proved.
Let us then extract the coefficients $(b_{k,n})_{k=1}^r$ for the mod-$\phi$ approximation scheme of order $r$: thus, we need to compute the power series expansion of $\psi_n(\xi)$ in powers of $z=\ee^{\I\xi}-1$.
We start by noticing that 
$$\widehat{\mu}_Z(\xi) = \sum_{j=0}^\infty \frac{1}{j!}\, \E{(Z)_j} (\ee^{\I\xi} - 1)^j,$$
where $(Z)_j := Z(Z-1)\cdots (Z-k+1)$ is the $j$-th falling factorial and $\E{(Z)_j}$ is therefore the $j$-th factorial moment of $Z$. As a consequence, the series in the exponential \eqref{eq:residue_mod-compound_poisson} rewrites as:
\begin{align*}
    &\sum_{k=1}^\infty \frac{(-1)^{k-1}}{k}\,\left(\sum_{j=1}^\infty \frac{1}{j!}\,\E{(Z)_j}\,z^j \right)^k\,\mathfrak{p}_{k,n} \\ &= \sum_{l=1}^\infty z^l\left(\sum_{k=1}^\infty \frac{(-1)^{k-1}}{k}\,\mathfrak{p}_{k,n}\,\sum_{\substack{l=j_1+\cdots+j_k \\ j_1,\ldots,j_k \geq 1}} \frac{\E{(Z)_{j_1}}\cdots \E{(Z)_{j_k}}}{(j_1)!\cdots (j_k)!}\right) \\ 
    &= \sum_{l=1}^\infty z^l\left(\sum_{k=1}^l (-1)^{k-1} (k-1)!\,\mathfrak{p}_{k,n}\,\sum_{\substack{l=\lambda_1+\cdots+\lambda_k \\ \lambda_1\geq \cdots \geq \lambda_k \geq 1}}  \frac{\E{(Z)_{\lambda_1}}\cdots \E{(Z)_{\lambda_k}}}{(m_1(\lambda))!\cdots (m_l(\lambda))!\,\,(\lambda_1)!\cdots (\lambda_k)!}\right).
\end{align*}
In these formul{\ae}, the sum on the second line runs over compositions $l=j_1+\cdots+j_k$ of size $l$ and length $k$ (sequences that sum to $l$), and the sum on the third line runs over integer partitions $l=\lambda_1 +\cdots + \lambda_k$ with $\lambda=(\lambda_1\geq \cdots \geq \lambda_k)$ (non-increasing sequences). For $i \leq l$, $m_i(\lambda)$ is the number of parts $\lambda_j$ equal to $i$. Similar combinatorial arguments are detailed in the proof of Theorem \ref{thm:exp_log}; see Appendix \ref{sec:mobius_function}. Now, since the number of set partitions $\pi \in \Pi(l)$ with sizes of blocks given by an integer partition $\lambda=(\lambda_1\geq \cdots \geq \lambda_k)$ is 
$$\frac{l!}{(m_1(\lambda))!\cdots (m_l(\lambda))!\,\lambda_1!\cdots \lambda_k!},$$
and since the Möbius function $\mu(\pi,\widehat{1}_l)$ of a set partition with $k=|\pi|$ blocks is $(-1)^{k-1}\,(k-1)!$, we can rewrite:
$$\sum_{k=1}^\infty \frac{(-1)^{k-1}}{k}\,\left(\sum_{j=1}^\infty \frac{1}{j!}\,\E{(Z)_j}\,z^j \right)^k\,\mathfrak{p}_{k,n} = \sum_{l=1}^\infty \frac{z^l}{l!}\left(\sum_{\pi \in \Pi(l)} \mu(\pi,\widehat{1}_l)\,\mathfrak{p}_{|\pi|,n}\left(\prod_{B \in \pi} \E{(Z)_{|B|}}\right) \right).$$
Substituting this expansion in Equation \eqref{eq:residue_mod-compound_poisson} and using the first part of Theorem \ref{thm:exp_log} with $g_l = \sum_{\pi \in \Pi(l)} \mu(\pi,\widehat{1}_l)\,\mathfrak{p}_{|\pi|,n} (\prod_{B \in \pi} \E{(Z)_{|B|}})$, we get:
\begin{align*}
b_{k,n} &= \frac{1}{k!} \sum_{\sigma \in \Pi(k)} \prod_{B \in \sigma} g_{|B|} = \frac{1}{k!}\sum_{\sigma \in \Pi(k)} \prod_{B \in \sigma} \left(\sum_{\tau_B \in \Pi(|B|)} \mu(\tau_B,\widehat{1}_{|B|})\,\mathfrak{p}_{|\tau_B|,n} \left(\prod_{C \in \tau_B} \E{(Z)_{|C|}}\right)\right).
\end{align*}
We notice that choosing one partition $\tau_B \in \Pi(|B|)$ for each block $B \in \sigma$ is equivalent to choosing a subpartition $\tau$ of $\sigma$, so the coefficient can be expressed more compactly as:  
\begin{equation*}
    b_{k,n} = \frac{1}{k!} \sum_{\sigma \in \Pi(k)} \sum_{\substack{\tau \in \Pi(k) \\ \tau \le \sigma}} \mu(\tau, \sigma) \prod_{D \in \tau} \E{(Z)_{|D|}} \prod_{B \in \sigma} \mathfrak{p}_{n^\sigma_\tau(B), n}
\end{equation*}
where $n^\sigma_\tau(B)$ denotes the number of blocks of $\tau$ contained in the block $B$ of $\sigma$; and where we have furthermore recognized the M\"obius function for the incidence algebra of the poset $\Pi(k)$, given by $\mu(\tau, \sigma) = (-1)^{|\tau| - |\sigma|} \prod_{B \in \sigma} (n^\sigma_\tau(B) - 1)!=\prod_{B \in \sigma} \mu(\tau_B,\widehat{1}_{|B|})$.
\end{proof}

The implementation of mod-compound Poisson approximation schemes for credit risk applications is numerically more problematic than in the case of mod-Poisson schemes. 
In particular, integrals with respect to the reference infinitely divisible law are more difficult to evaluate and closed-form expressions, such as formul{\ae} \eqref{eq:tail_function_approximation} and \eqref{eq:call_function_approximation}, are not available anymore, except in few very specialized cases. Nevertheless these integrals can be evaluated numerically (for instance by estimating the compound Poisson law via Panjer recursion) but, depending on the application at hand and the particular distribution of the random exposures, this might lead to estimations that are as computationally expensive as the recursive methodology.

\section{Application: Estimation of risk measures}
\label{sec:risk_measures_estimation}

\subsection{Background}

A commonly used risk measure for market and credit risk applications is the \emph{Value at Risk} ($\mathrm{VaR}$), defined as
$$\var{\alpha}{L_n} := \inf \left\{ t \in \Rb \st \P{L_n \le t} \ge \alpha \right\},$$
which quantifies the minimum capital required to cover all portfolio losses with a probability at least equal to $\alpha$. The parameter $\alpha$ is known as the \emph{confidence level} and higher values of this parameter correspond to higher and more stringent capital requirements. 
The definition of $\var{\alpha}{L_n}$ is mathematically equivalent to the generalized inverse of the distribution function of $L_n$, which is also known as the $\alpha$-quantile of $L_n$.\medskip 

The $\mathrm{VaR}$ is the most commonly used risk measure in financial practice, despite the fact that it is not a coherent risk measure \cite{artzner1999coherent}, which means that it fails to account for diversification effects when risk is aggregated across several portfolios. 
A risk measure that is commonly employed to solve this problem is the \emph{Expected Shortfall} (ES), also known as \emph{conditional $\mathrm{VaR}$}, which is defined as follows:
$$\es{\alpha}{L_n} := \frac{1}{1-\alpha} \int_\alpha^1 \var{u}{L_n} du = \E{L_n|L_n > \var{\alpha}{L_n}}.$$

One can think of $\es{\alpha}{L_n}$ at confidence level $\alpha$ as the expected value of portfolio losses, given that these losses already exceed $\var{\alpha}{L_n}$. One can show that $\es{\alpha}{L_n}$ is always greater or equal to $\var{\alpha}{L_n}$ for all confidence levels $\alpha$, so that $\mathrm{ES}$ is more conservative than $\mathrm{VaR}$. It is also a coherent risk measure, because it incorporates information about all potential losses, including the ones above the confidence level $\alpha$, and it is therefore the simplest modification of $\mathrm{VaR}$ that yields a theoretically acceptable risk measure. The computation of $\mathrm{ES}$ is more demanding than the computation of $\mathrm{VaR}$, because it requires the availability of accurate estimations for the entire tail function of the loss distribution: this is the reason why simulation-based methods, such as Monte Carlo simulation, typically perform poorly in estimations of portfolio $\mathrm{ES}$. 
\medskip 

We remark that if $L_n$ is a discrete random variable (this is typically the case if the distribution has been discretized for purposes of numerical evaluation) with $ N := \|L_n\|_\infty < \infty$ , then $\var{\alpha}{L_n}$ is integer-valued and is given by the following left-continuous step function:
\begin{equation}
\label{eq:var}
\var{\alpha}{L_n} = \sum_{k=0}^N k \: \1_{(\P{L_n \le k-1}, \P{L_n \le k}]}(\alpha),    
\end{equation}
while $\es{\alpha}{L_n}$ is given by
\begin{equation}
\label{eq:es}
\es{\alpha}{L_n} = \frac{1}{1-\alpha} \left( \left( \P{L_n \le \var{\alpha}{L_n}} - \alpha \right)\var{\alpha}{L_n}  \,\,+\!\!\!\!  \sum_{k > \var{\alpha}{L_n} }^N \!\!\!k \: \P{L_n = k} \right).
\end{equation}

\subsection{Estimation or risk measures}

It is clear from Equations \eqref{eq:var} and \eqref{eq:es} that the estimation of the $\mathrm{VaR}$ and the $\mathrm{ES}$ depends crucially on accurate estimates for the tail function $\P{L_n \ge x}$ of portfolio losses. It is therefore instructive to compare estimation methods first of all on the task of tail function estimation for a representative credit portfolio model, in our case a single-factor Gaussian copula with $n=250$ obligors, heterogeneous average default probabilities uniformly distributed in $[2\%, 8\%]$ and equicorrelation parameter $\rho = 0.3$. The results do not vary qualitatively for different choices of the parameters.
\begin{center}
\begin{figure}[ht]
    \includegraphics[width=0.9\textwidth]{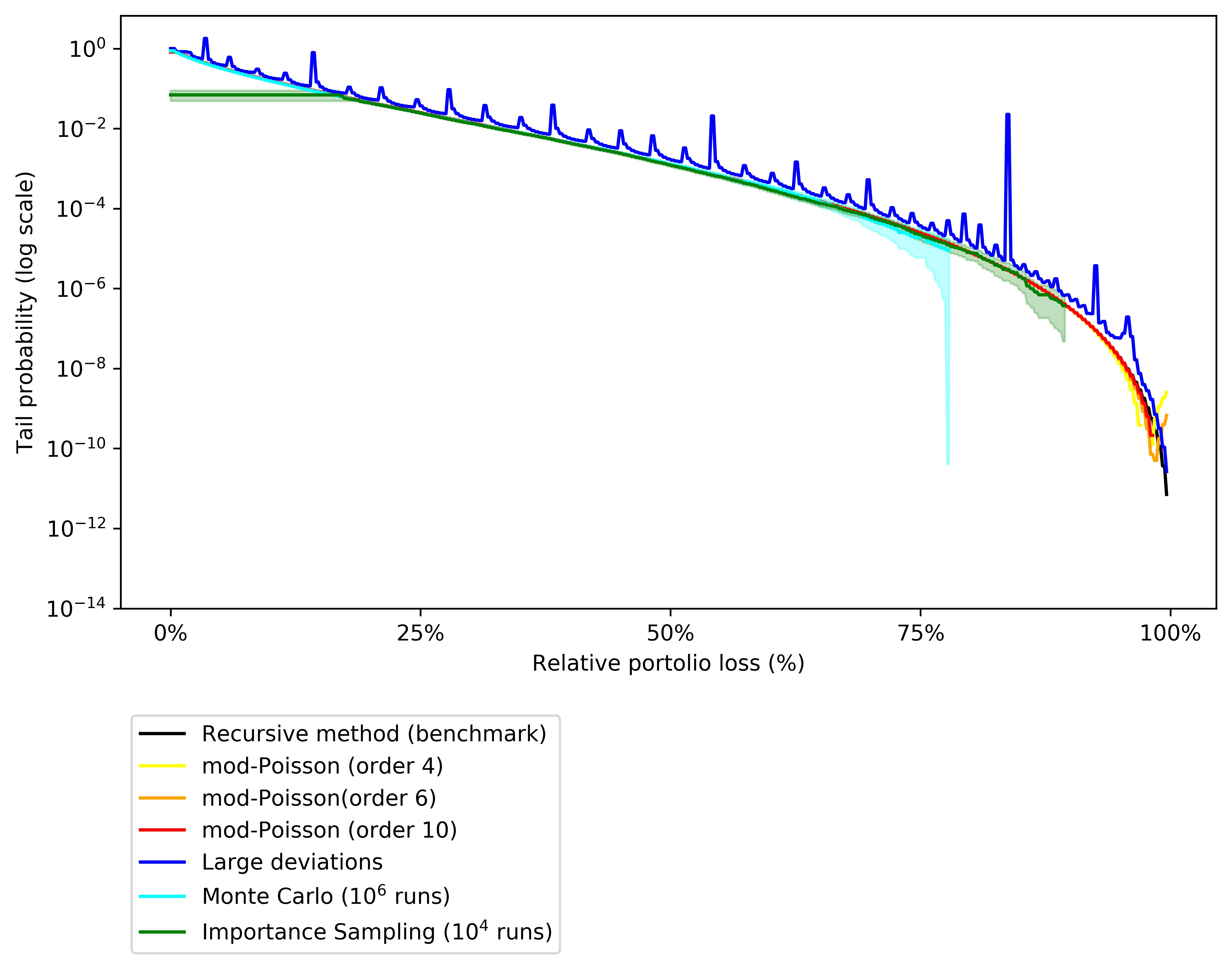}
    \caption{Estimated tail functions on a logarithmic scale. For simulation-based methods (i.e.~Monte Carlo and Importance Sampling) mean estimates are reported within their 99\% asymptotic confidence interval (shaded areas).}
    \label{fig:tail_function}
\end{figure}
\end{center}
Figure \ref{fig:tail_function} shows a comparison of the estimated tail functions for an exact method (the recursive method), two semi-analytical methods (mod-Poisson with varying order and large deviations approximation) and two simulation-based methods (Monte Carlo and importance sampling). The reader is referred to Appendix \ref{sec:estimation_overview} for a self-contained presentation of all these estimation models, together with full details of their numerical implementation.

\begin{itemize}
    \item The recursive method can be used as a benchmark to assess the accuracy of the other methods, because it is an exact procedure for the computation of the loss distribution, up to the numerical integration error due to integration over the portfolio mixing variable and any rounding errors due to finite machine precision, both of which are in practice of order $10^{-15}$. 
A better measure of the performance of each estimation method can be obtained by looking at the signed relative errors of the estimated tail probabilities (computed with respect to the benchmark) as a function of the tail point, as plotted in Figure \ref{fig:tail_errors}.

\begin{center}
\begin{figure}[ht]
    \includegraphics[width=0.99\textwidth]{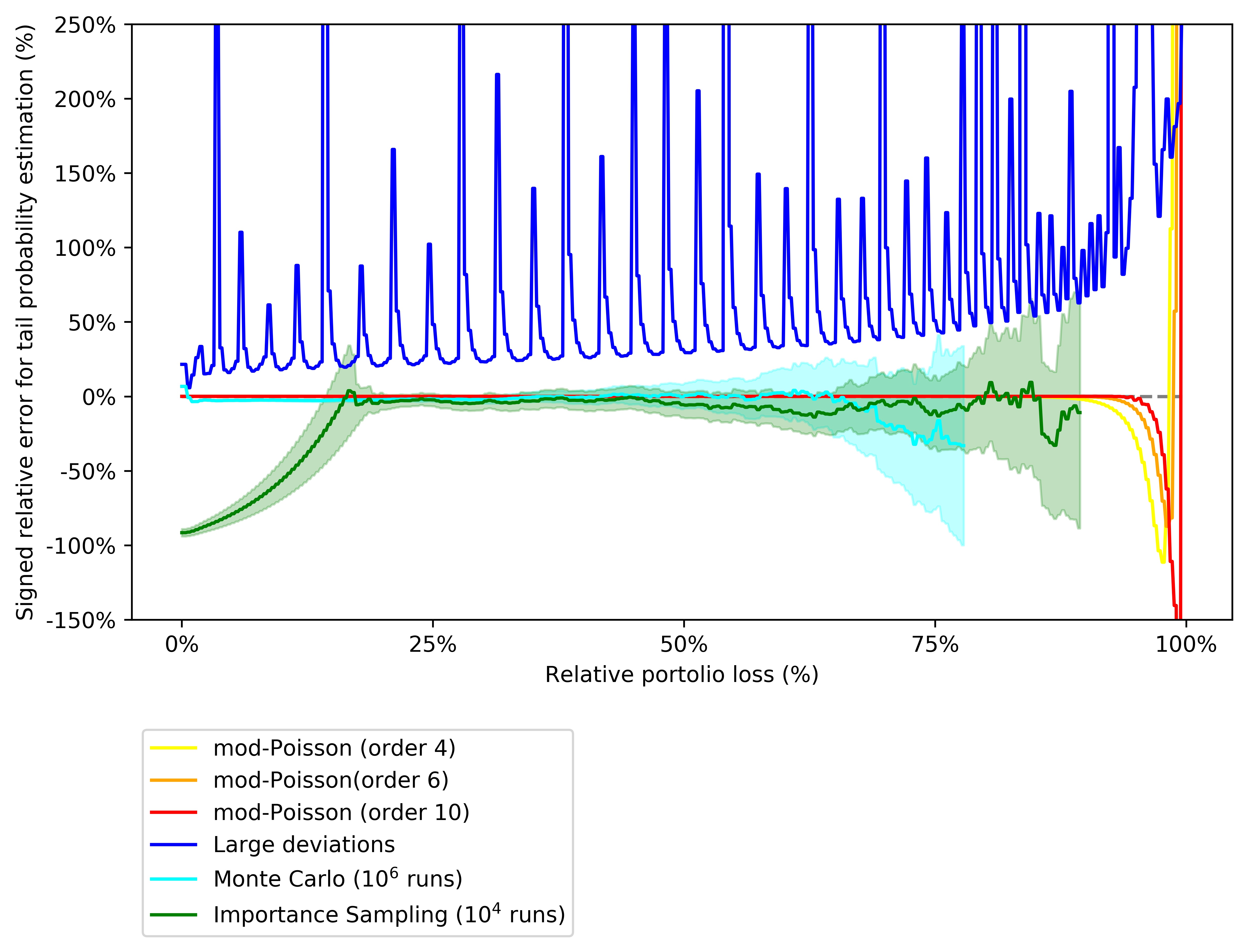}
    \caption{Signed relative errors of estimated tail probabilities for all levels of relative portfolio losses. Benchmark value (i.e.~the assumed true value with respect to which errors are computed) computed via the recursive method. For simulation-based methods (i.e.~Monte Carlo and Importance Sampling) mean estimates are reported within their 99\% asymptotic confidence interval (shaded areas).}
    \label{fig:tail_errors}
\end{figure}
\end{center}

\item The plain Monte Carlo method has been implemented as in Algorithm \ref{alg:monte_carlo_simulation} with $10^6$ simulation runs. Due to the finite number of simulations, it is able to estimate the probability of relative losses only up to the $75\%$ level and with increasing uncertainty, as shown by the widening of the asymptotic confidence intervals around the mean estimates. The mean estimates themselves show a negative bias, due to the undersampling of rare, large losses. Higher accuracy can be obtained by suitably increasing the number of simulations at the price of higher 
computational times.

\item The importance sampling method, by addressing the rare-event simulation problem as explained in Section \ref{subsec:importance_sampling}, is able to estimate tail probabilities at higher levels of relative losses with only $10^4$ simulation runs. Nevertheless, this method is also characterized by widening confidence intervals, with the result that only estimates close to $85\%$ of relative losses can be conceivably used with any statistical confidence. We further remark that higher uncertainty is visible in Figure \ref{fig:tail_errors} not only for high levels of relative losses, but also for low ones. This is because after performing the exponential tilting described in Algorithm \ref{alg:importance_sampling_simulation_second_step} the mean of the loss distribution has been shifted and low losses have become rare events. In practice one could fix this issue by gluing together the estimates coming from a plain Monte Carlo method for low levels of losses and the estimates from the importance sampling method for high levels, as 
suggested in \cite{glasserman2005importance}.

\item The large deviations approximation is characterized by an erratic behavior. This is partly due to the discreteness of the distribution $L_n$, but also to the numerical integration over the mixing variable of the Gaussian copula, which is unstable, possibly because of the denominator in Equation \eqref{eq:large_deviations_estimator}. In general one can expect this method to perform well only in the asymptotic regime as the number of obligors becomes large, but for finite portfolios of hundreds of obligors the performance is disappointing and the estimates are biased and affected by large relative errors.

\item The mod-Poisson approximation schemes show a remarkable accuracy for most of the tail function, except at the level of relative losses higher than $90\%$, i.e.~losses of probability lower than $10^{-8}$ for our representative portfolio, as can be deduced from Figure \ref{fig:tail_function}). Above that level the approximation order becomes an important tuning parameter: approximation schemes with higher order are able to maintain higher accuracy farther in the tail. Nevertheless, the performance at low approximation orders -- such as $r = 4, 6, 10$ as shown in Figure \ref{fig:tail_errors} -- is already very satisfactory, since mod-Poisson approximation schemes show high relative errors only on probabilities of order smaller than $10^{-8}$, whose contribution to the estimation of risk measures is negligible, as will be shown next.
\end{itemize}
\medskip

Table \ref{table:var_es} compares the estimates of the $\mathrm{VaR}$ and $\mathrm{ES}$ (obtained using Equations \eqref{eq:var} and \eqref{eq:es}) of the representative portfolio for a selection of typically used confidence levels, ranging from $95\%$ to $99.9999\%$, while Figure \ref{fig:var_errors} shows the signed relative errors in the estimation of $\mathrm{VaR}$ as a function of the confidence levels on a logarithmic scale. \medskip

The best performing methods are clearly the mod-Poisson approximation schemes, which yield estimates that are almost always identical to the benchmark value for any confidence value, already for order $r=4$. The other methods are characterized by diminishing accuracy at higher confidence levels and increasing uncertainty in the estimates, in the case of simulation-based methods.

\begin{landscape}
\thispagestyle{empty}
\begin{flushleft}
\begin{table}
\hspace*{-3cm}
\begin{tabular}{c|c|cccccc} 
 \hline
 $\mathrm{VaR}$ level & \begin{tabular}{c} Benchmark \\ (recursive) \end{tabular} & \begin{tabular}{c} Large \\ deviations \end{tabular} & \begin{tabular}{c} Monte Carlo \\ ($10^6$ runs) \\\relax Mean [$99\%$ CI] \end{tabular} & \begin{tabular}{c} Importance \\ Sampling \\ ($10^4$ runs) \\\relax Mean [$99\%$ CI]  \end{tabular} & \begin{tabular}{c} Mod-Poisson \\(order=4) \end{tabular} & \begin{tabular}{c} Mod-Poisson \\(order=6) \end{tabular} & \begin{tabular}{c} Mod-Poisson \\(order=10) \end{tabular} \\ \hline 
 95\% & 48 & 54 & 47 & 47 & 48 & 48 & 48 \\ \hline
 99\% & 82 & 89 & 82 [81, 83] & 81 [80, 83] & 82 & 82 & 82 \\ \hline
 99.99\% & 169 & 173 & 168 [164, 172] & 168 [165, 170] & 169 & 169 & 169 \\ \hline
 99.9999\% & 218 & 222 & n.a. & 215 [212, 221] & 218 & 218 & 218 \\ \hline
\end{tabular}
\vspace*{1cm}
\hspace*{-3cm}
\begin{tabular}{c|c|cccccc} 
 \hline
 $\mathrm{ES}$ level & \begin{tabular}{c} Benchmark \\ (recursive) \end{tabular} & \begin{tabular}{c} Large \\ deviations \end{tabular} & \begin{tabular}{c} Monte Carlo \\ ($10^6$ runs) \\\relax Mean [$99\%$ CI] \end{tabular} & \begin{tabular}{c} Importance \\ Sampling \\ ($10^4$ runs) \\\relax Mean [$99\%$ CI]  \end{tabular} & \begin{tabular}{c} Mod-Poisson \\(order=4) \end{tabular} & \begin{tabular}{c} Mod-Poisson \\(order=6) \end{tabular} & \begin{tabular}{c} Mod-Poisson \\(order=10) \end{tabular} \\ \hline 
 95\% & 68.91 & 79.08 & 68.50 [67.99, 69.01] & 68.16 [66.06, 69.28] & 68.91 & 68.91 & 68.91 \\ \hline
 99\% & 103.06 & 118.92 & 102.81 [101.85, 103.76] & 102.38 [100.86, 103.84] & 103.06 & 103.06 & 103.06 \\ \hline
 99.99\% & 181.65 & 198.71 & 160.70 [147.39, 172.90] & 179.99 [177.49, 182.09] & 181.64 & 181.65 & 181.65 \\ \hline
 99.9999\% & 223.73 & 231.13 & n.a. & n.a. & 223.57 & 223.70 & 223.73 \\ \hline
\end{tabular}
\caption{Estimates for $\mathrm{VaR}$ (above) and $\mathrm{ES}$ (below) across estimation methods for a selection of confidence levels (95\%, 97.5\%, 99\%, and 99.99\%). For simulation-based methods, such as Monte Carlo and Importance Sampling, the mean estimate is shown together with its 95\% asymptotic confidence interval in square brackets.}
\label{table:var_es}
\end{table}
\end{flushleft}
\end{landscape}

\begin{center}
\begin{figure}[ht]
    \includegraphics[width=\textwidth]{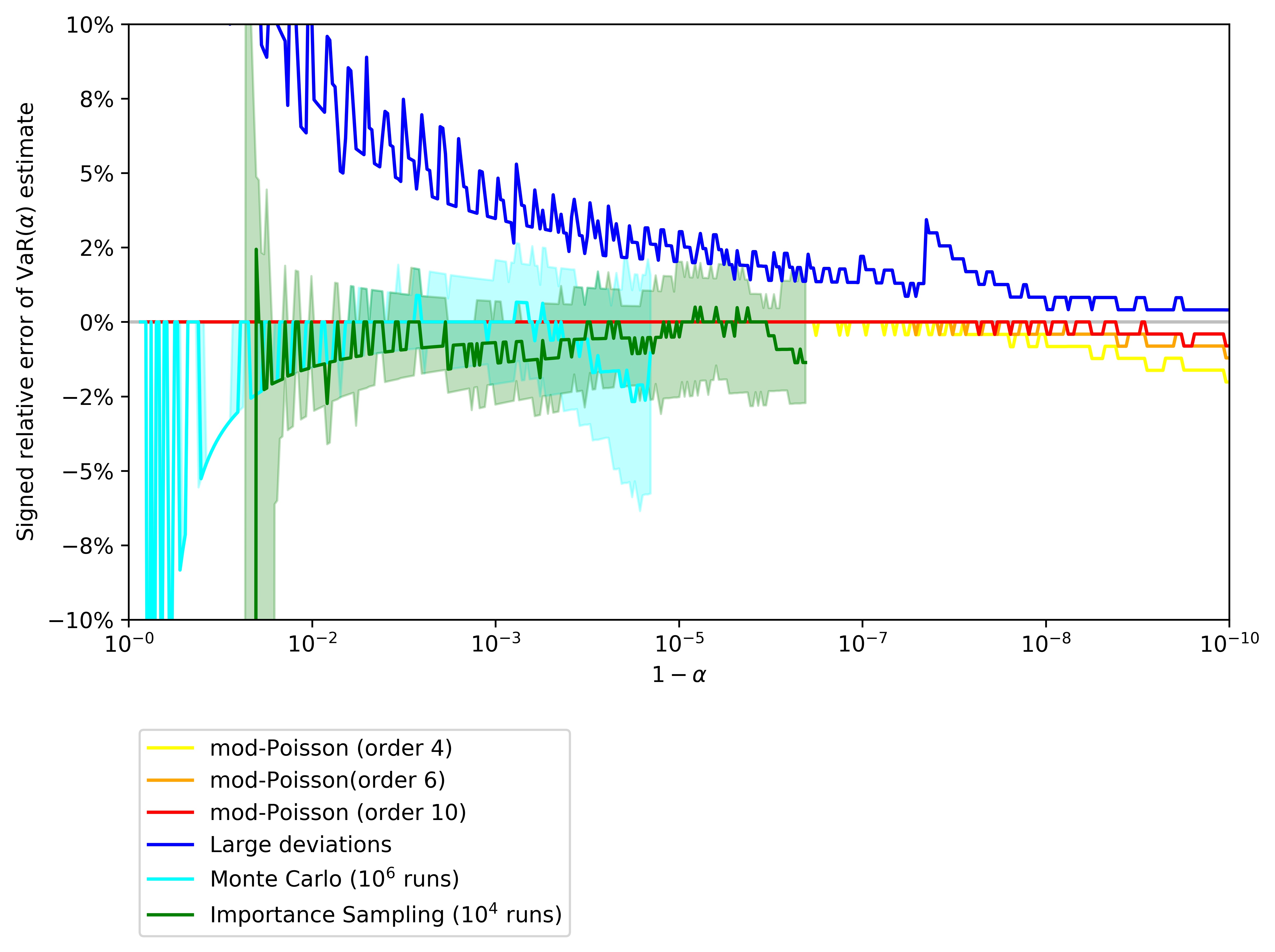}
    \caption{Signed relative errors of estimated $\mathrm{VaR}$ for confidence levels on a logarithmic scale. Benchmark value (i.e.~the assumed true value with respect to which errors are computed) computed via the recursive method. For simulation-based methods (i.e.~Monte Carlo and Importance Sampling) mean estimates are reported within their 99\% asymptotic confidence interval (shaded areas).}
    \label{fig:var_errors}
\end{figure}
\end{center}

\subsection{Computational time}
\label{subsec:computational_time}

The performance of an estimation method in terms of computational time can be particularly important for certain financial applications. In this section we compare empirically the performance of various methods for the estimation of tail probabilities and discuss relative advantages and disadvantages. All empirical tests are performed on a representative portfolio model, specifically a single-factor Gaussian copula with average default probabilities uniformly distributed in $[2\%, 8\%]$ and equicorrelation $\rho = 0.3$.

\begin{center}
\begin{figure}[ht]
    \includegraphics[width=\textwidth]{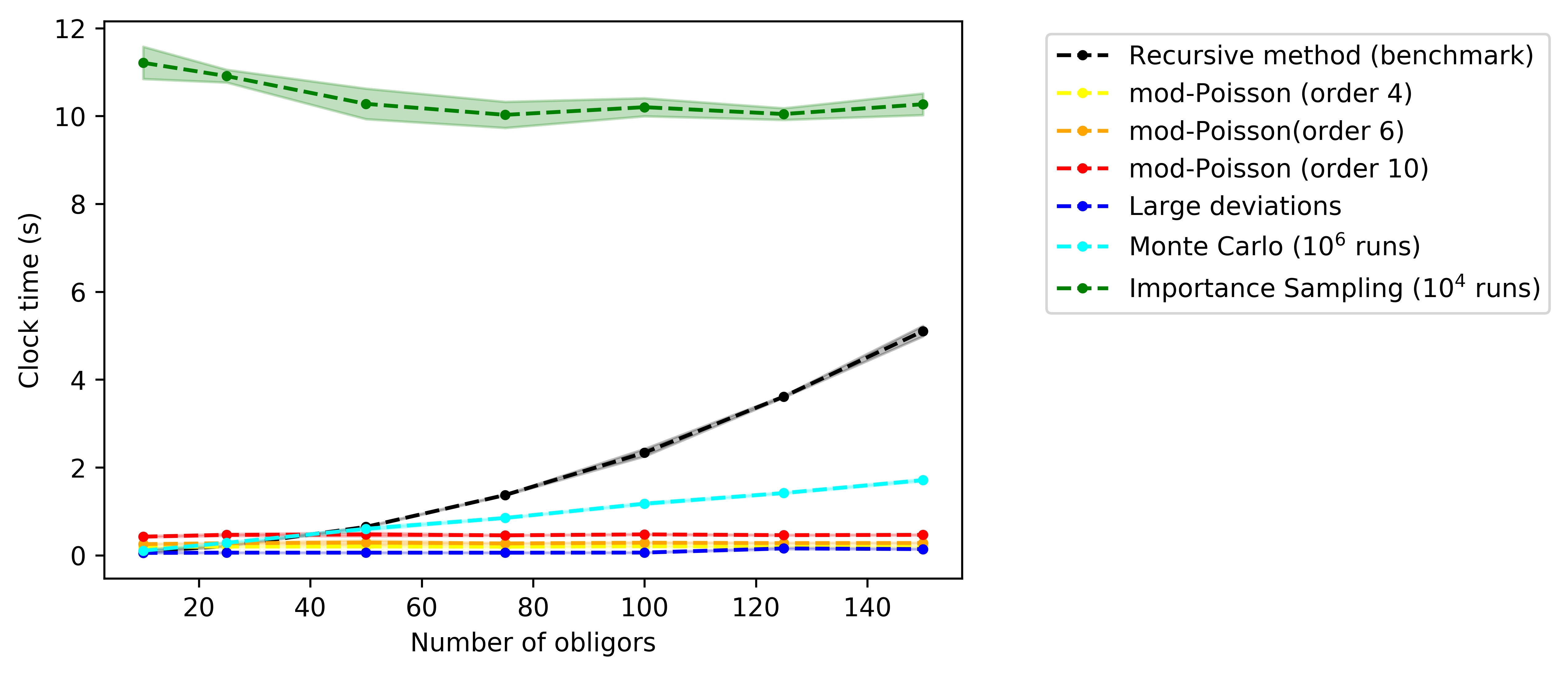}
    \caption{CPU clock time for estimation of a single tail probability for a selection of estimation methods as a function of the number of obligors. Shaded areas are $95\%$ asymptotic confidence intervals from $10$ runs of the estimation methods.}
    \label{fig:computational_time_tail}
\end{figure}
\end{center}

Figure \ref{fig:computational_time_tail} compares the elapsed CPU clock time for the estimation of a single tail probability of portfolio losses (i.e.~an evaluation of $\P{L_n > x}$ for a given value $x \ge 0$) as a function of the number of obligors.

\begin{itemize}
    \item The recursive methodology is the only known method for the exact estimation (up to numerical integration error) of the loss distribution, but its computational complexity scales quadratically in the number of obligors, as we already anticipated in the remarks following Algorithm \ref{alg:recursive_methodology} and as can be seen in Figure \ref{fig:computational_time_tail}. For this reason the recursive method is considered computationally expensive in the case of large portfolios (i.e.~with more than one hundred obligors) and approximate methods -- either semi-analytical or simulation-based -- are typically preferred.
\medskip

\noindent Furthermore, an important disadvantage of the recursive method comes from the fact that it is designed to output the full distribution of the portfolio losses, which might be wasteful in some applications. For instance, when computing risk measures it is necessary to compute the tail function of the loss distribution only at a few points -- more specifically in the part of the tail corresponding to high losses -- so that information of the full distribution is effectively useless. In contrast, semi-analytical methods, such as mod-Poisson approximation schemes or the large deviation approximation, yield approximations of the tail function at a single point and can thus be used to compute risk measures much more efficiently.
Another application for which the recursive methodology tends to be inefficient can be found in model risk management, where sensitivity analysis requires computing risk measures and other portfolio metrics repeatedly for a given credit risk model under slight perturbations of the model parameters. In this case the recursive methodology requires the expensive computation of a large number of very similar loss distributions, thus compounding the wastefulness issue discussed above.

\item 
As shown in Figure \ref{fig:computational_time_tail}, simulation-based methods -- such as Monte Carlo integration and importance sampling -- scale more favorably in the number of obligors. In particular, both methods scale linearly in the number of obligors, since they both require simulating a matrix of obligors' default indicators with number of columns equal to the number of obligors and number of rows equal to the number of simulation runs. \medskip

\noindent While for the Monte Carlo method this linear dependence is evident in Figure \ref{fig:computational_time_tail}, in the case of the importance sampling method it is concealed by the computational overhead stemming from the determination of the shifted mean $\mu$ in the first part of Algorithm \ref{alg:importance_sampling_simulation_second_step}. This preliminary optimization step turns out to be computationally expensive and makes the importance sampling algorithm unappealing for financial applications that require fast execution, such as product pricing for trading desks or the estimation of pre-trade risk for proprietary trading. Furthermore, it is worthwhile recalling that the importance sampling algorithm applies exclusively to the Gaussian factor copula model, so that its field of application is in any case already quite limited.
As far as the Monte Carlo method is concerned, a correct assessment of its time complexity must take into account the dependence on the number of simulations needed. While the importance sampling procedure can produce accurate estimations far into the tail with a low number of simulations (in practice of the order $10^4$), the plain Monte Carlo approach requires a much higher number of simulation runs, as explained in Section \ref{subsec:monte_carlo_simulation}. This is particularly problematic when high accuracy is required as in the computation of the Expected Shortfall of a portfolio and other distortion measures with non-zero spectrum at high quantile levels. 

\item Semi-analytical methods, such as mod-Poisson approximation schemes and the large deviations approximation, boast the best performance in terms of computational time, due to the fact that they require only the evaluation of known functions and a numerical integration over the copula factor, both of which can be performed efficiently.
To be more precise, these methods depend on the number of obligors only through the computation of a few coefficients. In the case of mod-Poisson approximation schemes it is necessary to compute the coefficients $\mathfrak{p}_{k,n}$ for $k = 2, \ldots, r$, where $r$ is the approximation scheme order, while in the case of the large deviations approximation it is only needed to evaluate the cumulant generating function $F$, its second derivative $F''$ and the optimal tilting $\lambda_x$. Nevertheless, these operations amount to computing specific functions of the vector of default probabilities $(p_1, p_2, \ldots, p_n)$ and are easily vectorized on any modern CPU, which results in a linear dependence on the number of obligors with a very small coefficient.
Indeed this linear dependence is empirically negligible for portfolios of even thousands of obligors and is effectively invisible in Figure \ref{fig:computational_time_tail}.

\end{itemize}

\begin{center}
\begin{figure}[ht]
    \includegraphics[width=0.8\textwidth]{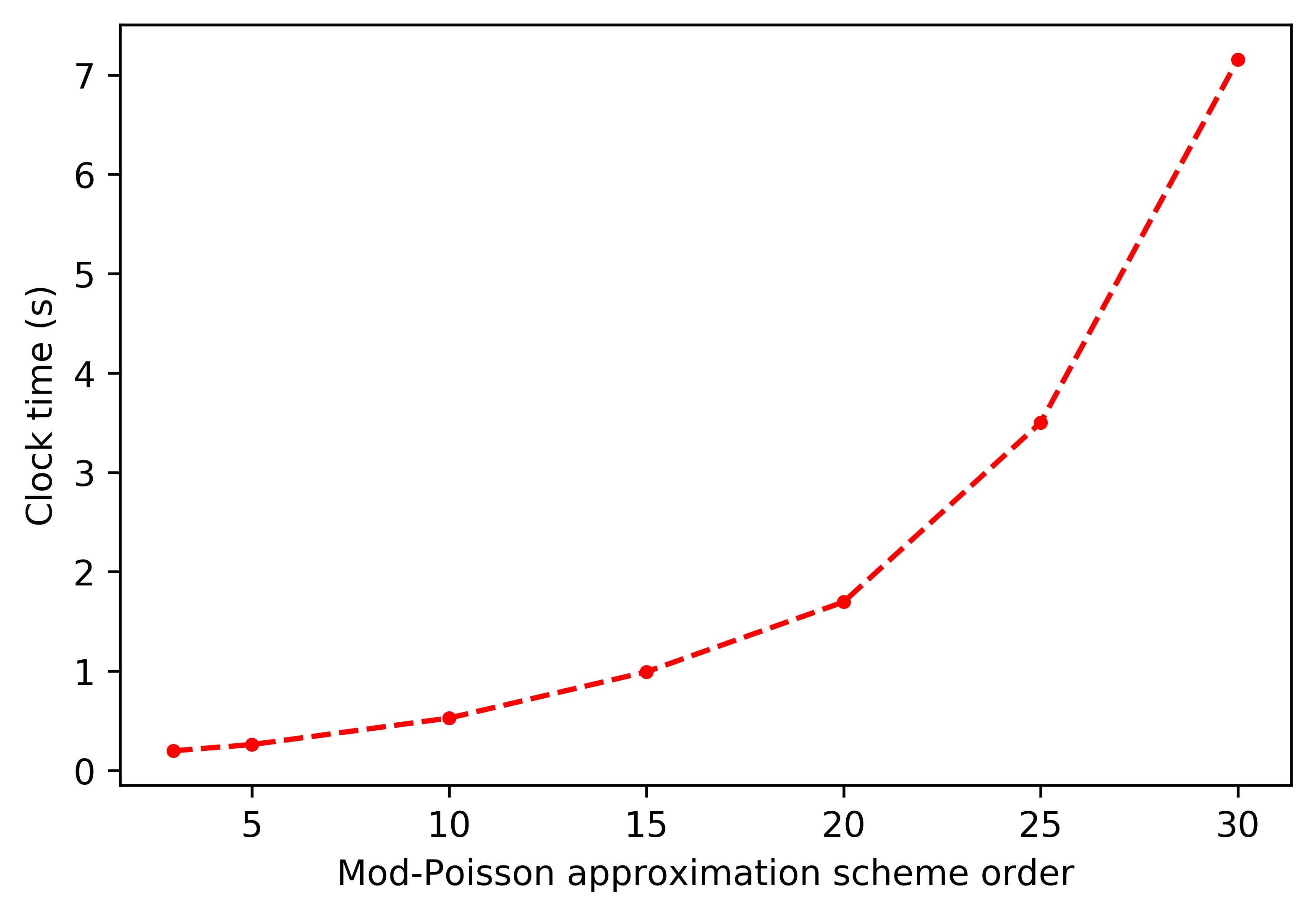}
    \caption{CPU clock time for estimation of a single tail probability for mod-Poisson approximation schemes as a function of the order.}
    \label{fig:computational_time_tail_order}
\end{figure}
\end{center}

The accuracy of mod-Poisson approximation schemes increases in the approximation order, as shown empirically in Figure \ref{fig:tail_errors}. It is therefore natural to investigate how the computational time burden increases when the order is increased.
From Theorem \ref{thm:modPoisson_convergence} and Remark \ref{rk:partitions} it is clear that the computation of the coefficient $b_{k,n}$ depends on $k$ as the number of integer partitions of $k$ with minimum block size $2$. In practice integer partitions can be generated efficiently by encoding them as ascending compositions -- rather than descending compositions, as is conventionally done -- and by exploiting some properties of this representation, as explained in \cite{kelleher2009generating}. One can then select only integer partitions with minimum block size $2$ by acceptance-rejection. This implementation yields an exponential time complexity in the square-root of $r$, as checked empirically for orders up to $r=30$ in Figure \ref{fig:computational_time_tail_order}.

\section{Application: CDO pricing}
\label{sec:cdo_pricing}

\subsection{Background}

A CDO (Collateralized Debt Obligation) is a financial product for the \emph{securitization} of credit portfolios, such as pools of residential mortgages or consumer loans. 
We refer the reader to \cite[Chapter 12]{mcneil2015quantitative} for an introduction to CDO pricing. Here we only mention that CDO pricing requires the valuation and comparison of payment cashflows, therefore the temporal evolution of the credit portfolio becomes important. Therefore we introduce the time dependence by denoting the total portfolio losses up to time $t$ as:
$$L_{t, n} 
= \sum_{i=1}^n Z_{i}\, Y_{t, i}.$$

Let us denote by $m$ the total number of tranches in the CDO, then the \emph{notional value} of the $j$-th tranche at time $t$ as a function of the underlying portfolio losses is given by:

$$ N_t^{(j)}(L_{t,n}) = \begin{cases} K_j - K_{j-1} & \text{if $L_{t,n} < K_{j-1}$}, \\ K_j - L_{t,n} & \text{if $K_{j-1} \le L_{t,n} \le K_j$}, \\ 0 & \text{if $L_{t,n} > K_j$}, \end{cases}$$ 
where $0 = K_0 < K_1 < \ldots < K_m$. The two values $K_{j-1}$ and $K_j$ are called the \emph{attachment} and \emph{detachment} points respectively. 
If the portfolio losses are below the attachment point, the tranche has a fixed value of $K_j - K_{j-1}$. As losses increase above that level, the tranche must absorb them and correspondingly loses value, up until the losses reach the detachment point and the tranche has become worthless.\medskip 

The notional value of tranches can also be expressed more compactly as follows:
\begin{equation}
\label{eq:tranche_notional}
N_t^{(j)} = (K_j - L_{t,n})^+ - (K_{j-1} - L_{t,n})^+, \quad j = 1, \ldots, m,
\end{equation}
which shows that CDO tranches have the same payoff as put spreads on the underlying credit portfolio. Similarly it is possible to define the cumulative \emph{tranche loss} up to time $t$, given by:
\begin{equation}
\label{eq:tranche_loss}
L_t^{(j)} = K_j - K_{j-1} - N_t^{(j)} = (L_{t,n} - K_{j-1})^+ - (L_{t,n} - K_j)^+, \quad j = 1, \ldots, m,
\end{equation}
which represents the losses incurred by the tranche and has the same payoff as a call spread on the underlying credit portfolio.
\medskip

In this section we are actually interested in pricing \emph{synthetic} CDOs, which involves the evaluation and comparison of the two cashflows, or \emph{legs}, of the counterparties: the premium payments leg and the default payments leg.
The CDO issuer makes premium payments at regular times, say $0 = t_0 < t_1 < \ldots < t_N = T$. Assuming time is measured in years, we can express these payments in terms of an annualized \emph{spread}, denoted by $s$. The premium payment at time $t_n$ from the $j$-th tranche is then equal to $s (t_n - t_{n-1}) N_{t_n}^{(j)}$, where $N_{t_n}^{(j)}$ is just the notional value of the tranche at time $t_n$, as given in Equation \eqref{eq:tranche_notional}.
In actual practice premium payments also include so called \emph{accrued payments}. More specifically, if an obligor in the reference portfolio defaults at a random time $T \in (t_{n-1}, t_n]$, then at time $t_n$  the CDO issuer is also required to pay the premium accrued over the time before the default occurred, i.e.~$s(T - t_{n-1})(L_{T}^{(j)} - L_{T-}^{(j)})$, where $L_{T}^{j}$ is the tranche loss at time $T$, as given in Equation \eqref{eq:tranche_loss}. By assuming a sufficiently thick time grid (i.e.~$N$ is sufficiently large), we can safely ignore accrued payments, as we will do in the following.\medskip

The total value at time $t=0$ of the premium cashflow can then be computed by taking the expectation of the discounted cashflow under an equivalent martingale measure, obtaining:
\begin{equation}
\label{eq:premium_leg}
L_{\mathrm{premium}}(s) = s \sum_{n=1}^N \ee^{-r t_n} (t_n - t_{n-1})\,\, \E{N_t^{(j)}},
\end{equation}
where $r$ denotes the deterministic risk-free interest rate. Notice that the uncertainty due to interest rate risk is many orders of magnitude smaller than the uncertainty due to default dependence. This is why incorporating interest rate risk in credit risk models typically leads to negligible contributions and the assumption of a deterministic interest rate is common in many credit risk applications.\medskip 

The CDO buyer makes default payments every time an obligor in the reference portfolio defaults. The discounted value at time $t=0$ of the default cashflow is given by the following integral:
$$ \int_0^T \ee^{-r t} dL_t^{(j)},$$
which is to be understood as a pathwise Riemann--Stieltjes integral, and which can be approximated as a stochastic Riemann sum over the premium payments time grid:
$$ \sum_{n=1}^N \ee^{-r t_n} (L_{t_n}^{(j)} - L_{t_{n-1}}^{(j)}).$$
The value at time $t=0$ of the default cashflow can then be computed, analogously to the premium case, as the expectation of the discounted cashflow under an equivalent martingale measure, yielding:
\begin{equation}
\label{eq:default_leg}
L_{\mathrm{default}} \approx \sum_{n=1}^N \ee^{-r t_n} \left( \E{L_{t_n}^{(j)}} - \E{L_{t_{n-1}}^{(j)}} \right).
\end{equation}
The fair value of the CDO can be deduced by equating the premium leg $L_{\mathrm{premium}}(s)$ in Equation $\eqref{eq:premium_leg}$ with the default leg $L_{\mathrm{default}}$ in Equation \eqref{eq:default_leg}. The value of the spread $s$ for which equality holds is:
$$ s = \frac{L_{\mathrm{default}}}{L_{\mathrm{premium}}(1)},$$
and is called the \emph{fair spread}. This is the quantity that is quoted in CDO exchanges and used to assess the relative cost a CDO contract.

\subsection{Estimation of call prices}

From the definitions of tranche notional value and tranche loss in Equations \eqref{eq:tranche_notional} and \eqref{eq:tranche_loss} it is clear that the computations of the two payment legs -- and therefore of the fair spread -- can be reduced to the problem of computing call and put spreads on the reference portfolio. Moreover, since the payoff of a put option can be expressed in terms of the payoff of a call option, this task can be further reduced to the accurate estimation of call prices only. This section is devoted to an empirical comparison of several estimation models on this particular task. All estimations refer to a representative credit portfolio model, a single-factor Gaussian copula with $n=100$ obligors, equicorrelation parameter $\rho = 0.1$ and heterogeneous average default probabilities sampled from a log-normal distribution with varying mean $p$ and standard deviation $\sigma=0.2$. The choice of a log-normal distribution is of course arbitrary and is done in analogy to (and to ease comparison with) the numerical experiments presented in \cite{el2008gauss, el2009stein}. The results do not vary qualitatively for different choices of the parameters.

\begin{figure}[ht]
    \centering
    \subfloat{{\includegraphics[width=0.5\textwidth]{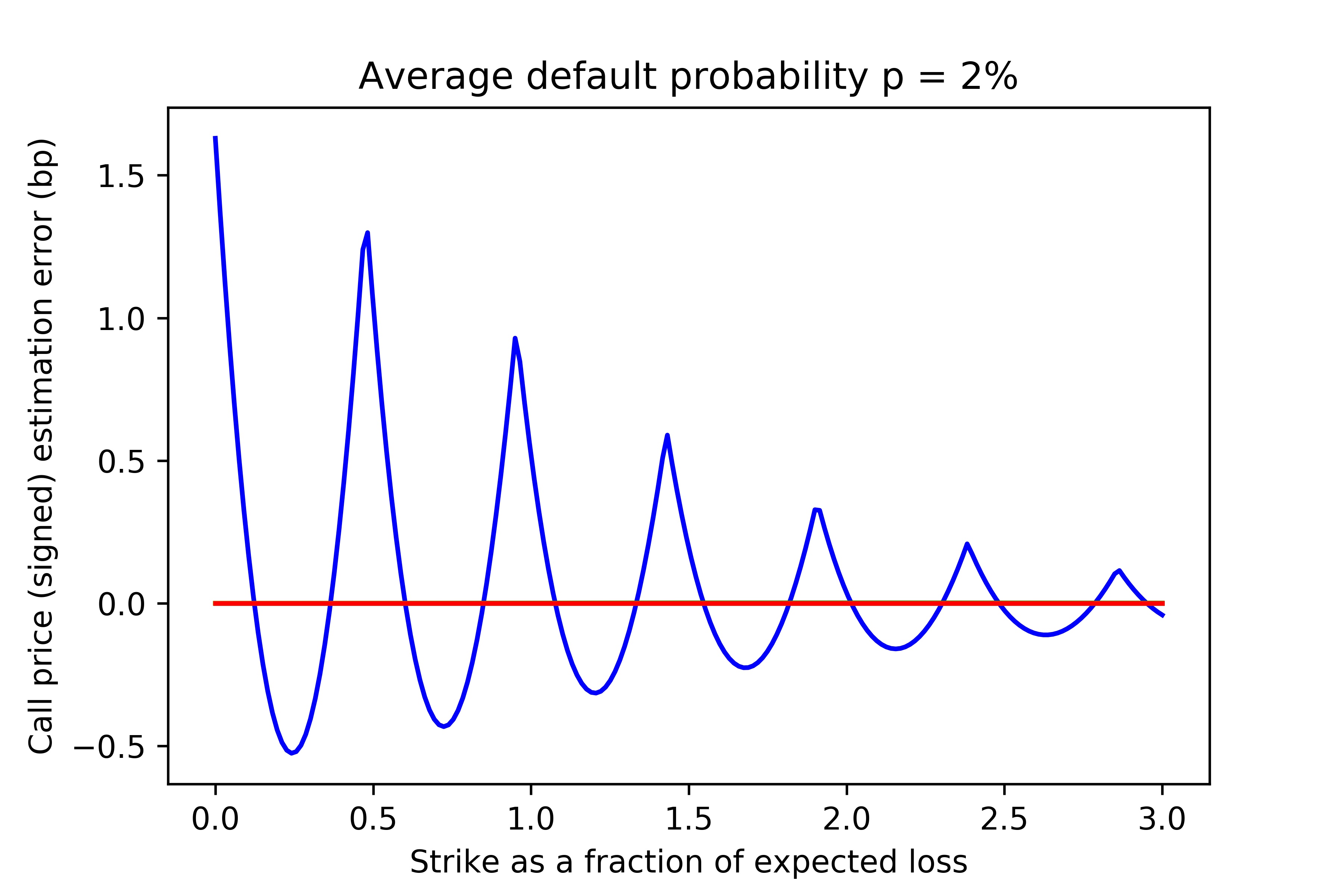}}}%
    \hspace*{-0.5cm}
    \subfloat{{\includegraphics[width=0.5\textwidth]{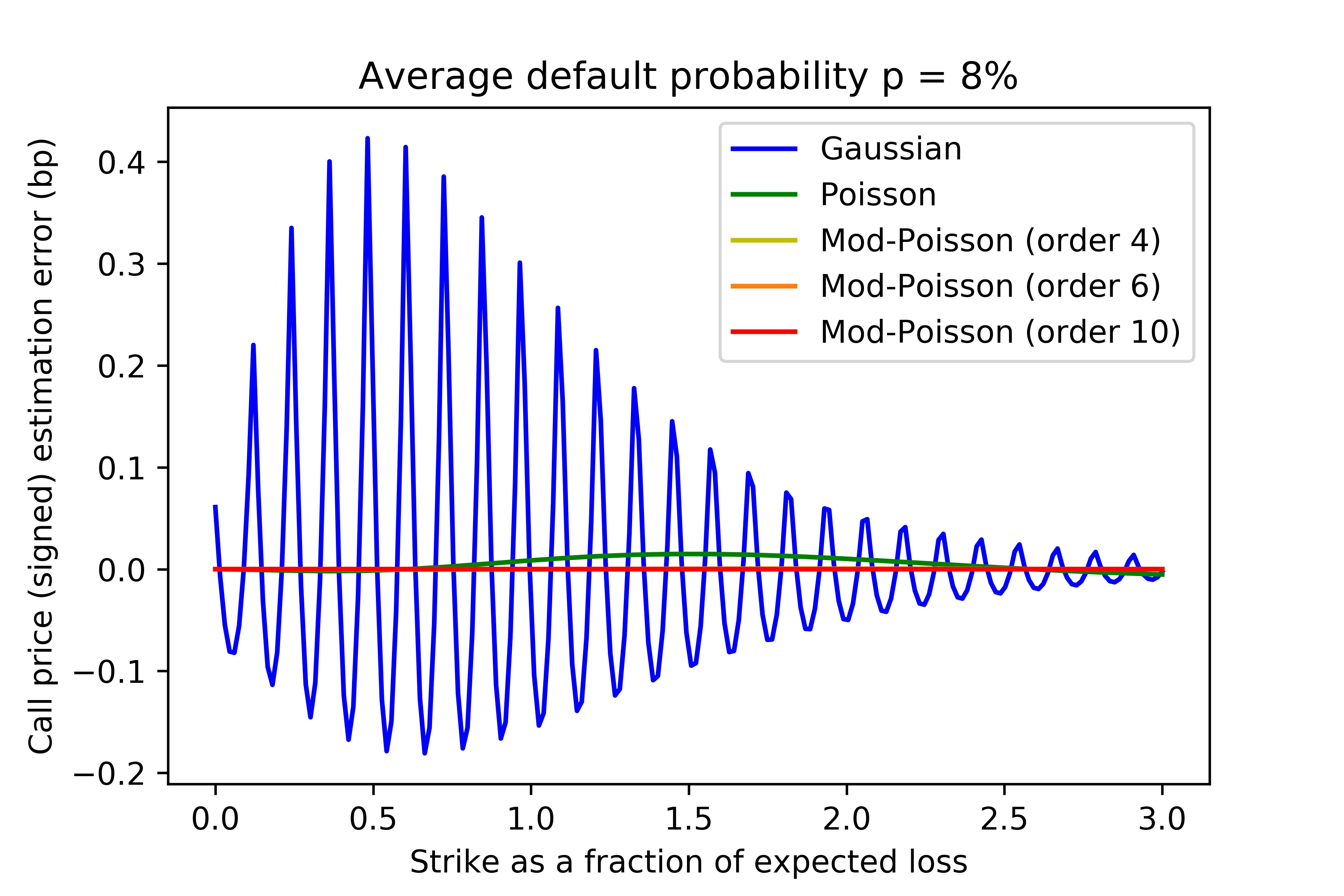} }}%
    \centering \\
    \subfloat{{\includegraphics[width=0.5\textwidth]{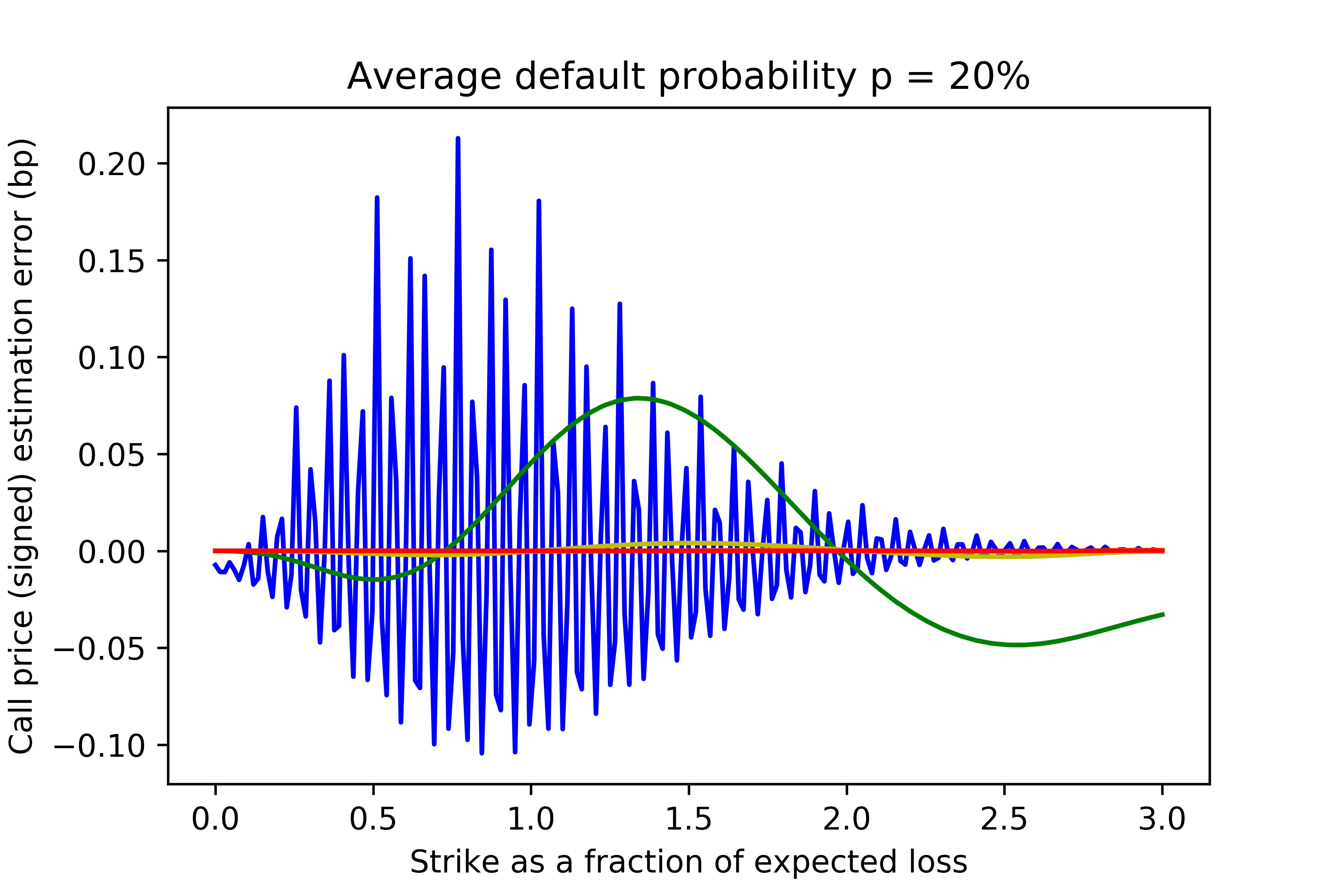}}}%
    \hspace*{-0.5cm}
    \subfloat{{\includegraphics[width=0.5\textwidth]{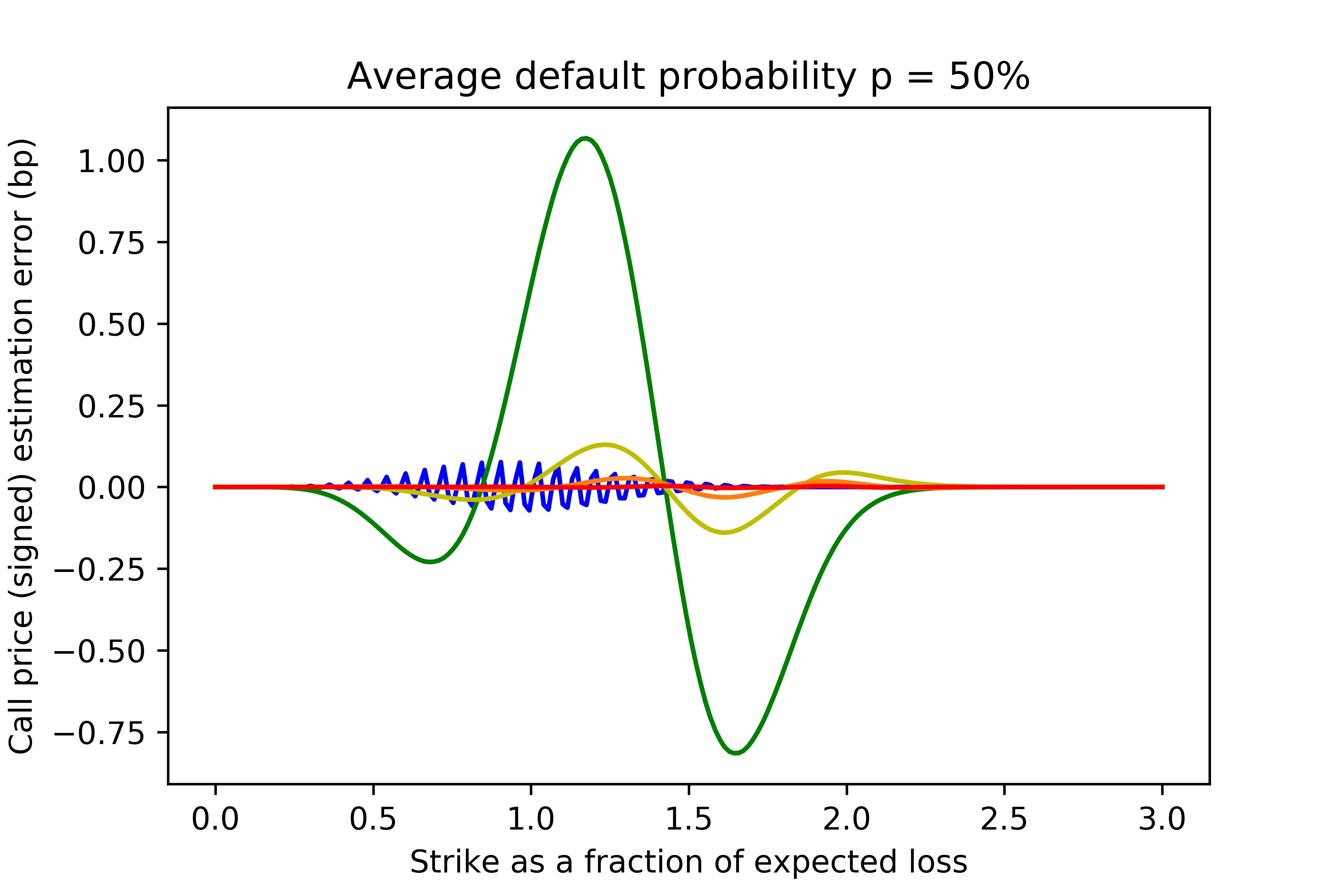} }}%
    \caption{Signed relative error for call price estimation measured in basis points (bp) as a function of the strike for different estimation methods and for increasing values of the average default probability in the portfolio. Method ``Poisson'' denotes the Chen--Stein first order correction to the Poisson approximation (which corresponds also to the mod-Poisson approximation scheme of order 2), while the method ``Gaussian'' denotes the Chen--Stein first order correction to the Gaussian approximation (see \cite{el2008gauss, el2009stein} or Appendix \ref{sec:estimation_overview}).}%
    \label{fig:calls}
\end{figure}
\medskip

Figure \ref{fig:calls} shows the signed relative error for estimated call prices measured in basis points (bp) as a function of the strike. Each subfigure refers to a different choice of the mean $p$ of the log-normal distribution from which the average single-obligor default probabilities are sampled.
The value of $p$ is the correct parameter to study the performance of these estimation methods, because as it increases the total portfolio losses move from a Poisson regime to a Gaussian one, correspondingly affecting the performance of the methods.
\begin{itemize}
\item For low values of $p$ all methods based on the Poisson approximation perform best, with negligible errors, while the Stein first-order correction to the Gaussian approximation yields comparatively larger errors, further characterized by an oscillation in the strike due to the approximation to a discrete distribution. 

\item As the default probability increases, the error associated with the Chen--Stein first-order correction to the Poisson approximation increases, while the Gaussian approximation yields better and better estimates. For this reason the authors in \cite{el2008gauss, el2009stein} (where the Chen--Stein's method is first applied to CDO pricing) propose a hybrid estimation method, in which they suggest to use either the first-order correction to the Gaussian approximation or to the Poisson one, depending on the specific value of $p$. 

\item As shown in Figure \ref{fig:calls}, mod-Poisson approximation schemes at higher order perform very well even in a Gaussian regime. For instance, in the case of $p = 50\%$ -- which is an extremely high average default probability in any conceivable credit risk setting -- the mod-Poisson approximation schemes of order $6$ and $10$ perform better than the first-order correction to the Gaussian approximation, so that the domain of validity of the Poisson approximation, if properly corrected, includes all credit risk applications, without the need for a hybrid method.

\end{itemize}

\subsection{Default leg, premium leg, and fair spread}

Accurate estimates of call prices can then be used in Equations \eqref{eq:default_leg} and \eqref{eq:premium_leg} for the determination of the default and premium legs of the CDO and finally for the estimation of the fair spread, which is effectively used to price a CDO contract.
Table \ref{table:cdo} compares the estimation of default legs, premium legs and fair spreads for five standardized tranches of a CDO written on a representative portfolio.

\begin{center}
\begin{figure}[ht]
    \includegraphics[width=\textwidth]{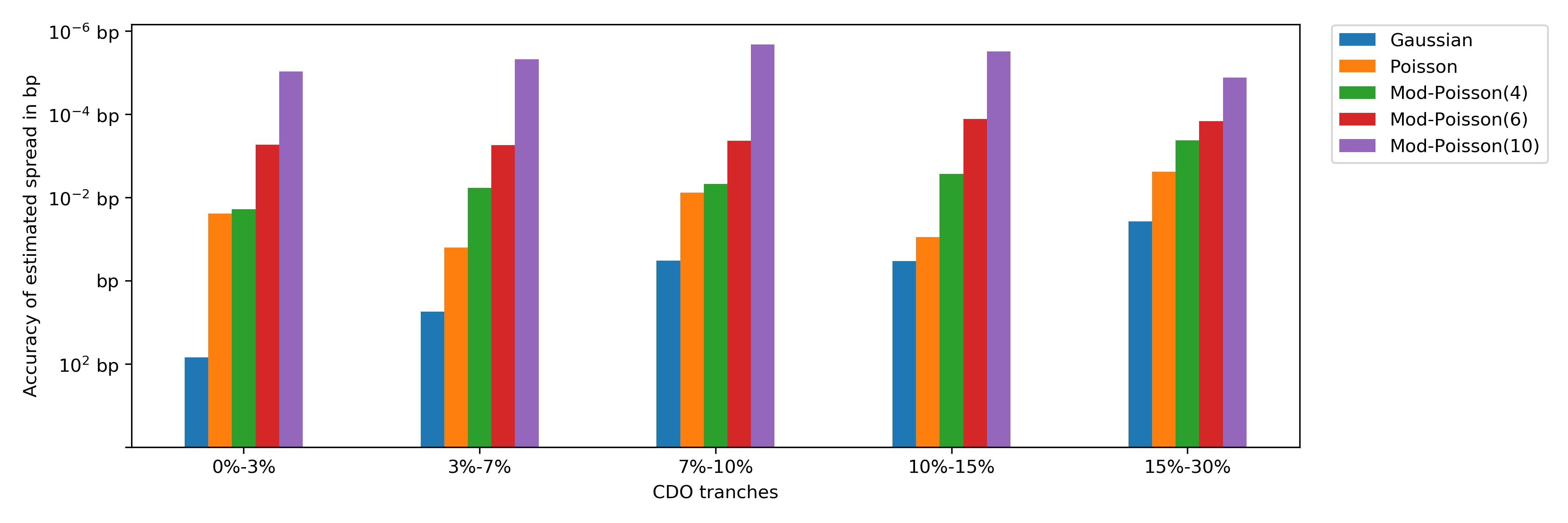}
    \caption{Decimal place accuracy of fair spread estimates for all estimation methods by tranche.}
    \label{fig:decimal_accuracy}
\end{figure}
\end{center}

The accuracy of the models' estimates can be more readily compared from Figure \ref{fig:decimal_accuracy}, where the decimal place accuracy of the fair spread estimates in Table \ref{table:cdo} are compared. The mod-Poisson approximation schemes yield the most accurate estimates, with exponentially improving precision as the approximation scheme order increases.

\clearpage
\begin{landscape}
\thispagestyle{empty}
\begin{center}
\begin{table}
\vspace*{-1cm}
\hspace*{-2cm}
\begin{tabular}{cc|r|rrrrr} 
 \hline
 \begin{tabular}{c} Attachment \\ points \end{tabular} &  & \begin{tabular}{c} Benchmark \\ (recursive) \end{tabular} & \begin{tabular}{c} Gaussian \\ approximation \end{tabular} & \begin{tabular}{c} Poisson \\ approximation \end{tabular} & \begin{tabular}{c} Mod-Poisson \\(order=4) \end{tabular} & \begin{tabular}{c} Mod-Poisson \\(order=6) \end{tabular} & \begin{tabular}{c} Mod-Poisson \\(order=10) \end{tabular} \\ \hline 
& Default leg & 232.5975 bp & 228.8759 bp & 232.5996 bp & 232.5979 bp & 232.5974 bp & 232.5975 bp\\
0\% - 3\% & Premium leg & 452.2145 bp & 451.0626 bp & 452.2208 bp & 452.2137 bp & 452.2145 bp & 452.2145 bp\\
& Fair spread & 5143.5210 bp & 5074.1488 bp & 5143.4961 bp & 5143.5404 bp & 5143.5204 bp & 5143.5210 bp\\ \hline
& Default leg & 200.2722 bp & 200.7338 bp & 200.2540 bp & 200.2716 bp & 200.2723 bp & 200.2722 bp\\
3\% - 7\% & Premium leg & 1364.6971 bp & 1362.7014 bp & 1364.7217 bp & 1364.6987 bp & 1364.6971 bp & 1364.6971 bp\\
& Fair spread & 1467.5213 bp & 1473.0575 bp & 1467.3613 bp & 1467.5153 bp & 1467.5218 bp & 1467.5213 bp\\ \hline
& Default leg & 62.8105 bp & 62.7749 bp & 62.8088 bp & 62.8099 bp & 62.8104 bp & 62.8105 bp\\
7\% - 10\% & Premium leg & 1248.7606 bp & 1248.8878 bp & 1248.7468 bp & 1248.7608 bp & 1248.7606 bp & 1248.7606 bp\\
& Fair spread & 502.9824 bp & 502.6464 bp & 502.9747 bp & 502.9777 bp & 502.9820 bp & 502.9825 bp\\ \hline
& Default leg & 33.6304 bp & 33.5575 bp & 33.6500 bp & 33.6310 bp & 33.6304 bp & 33.6304 bp\\
10\% - 15\% & Premium leg & 2204.4540 bp & 2204.5755 bp & 2204.4246 bp & 2204.4529 bp & 2204.4540 bp & 2204.4540 bp\\
& Fair spread & 152.5566 bp & 152.2176 bp & 152.6473 bp & 152.5594 bp & 152.5565 bp & 152.5566 bp\\ \hline
& Default leg & 7.2444 bp & 7.2698 bp & 7.2461 bp & 7.2447 bp & 7.2445 bp & 7.2444 bp\\
15\% - 30\% & Premium leg & 6738.6074 bp & 6738.5758 bp & 6738.6165 bp & 6738.6076 bp & 6738.6073 bp & 6738.6074 bp\\
& Fair spread & 10.7506 bp & 10.7883 bp & 10.7531 bp & 10.7511 bp & 10.7508 bp & 10.7507 bp\\ \hline
\end{tabular}
\vspace*{0.1cm}
\caption{Default leg, premium leg and fair spread for five tranches computed using different techniques. Benchmark values are exact and computed using the recursive methodology.}
\label{table:cdo}
\end{table}
\end{center}
\end{landscape}

\section{Conclusions} 
\label{sec:conclusions}

In this paper we have introduced mod-Poisson approximation schemes for the semi-analytical estimation of functionals of factor credit portfolio models. This technique is based on the theory of mod-$\phi$ convergence and mod-$\phi$ approximation schemes and relies on the computation of higher-order correction terms to the classic Poisson approximation. We also show how to extend the method to credit models with stochastic exposures using mod-compound Poisson approximation schemes.
\medskip

The method has been compared empirically with the recursive method, the large deviations approximation, the Chen--Stein method and the Monte Carlo simulation technique (with and without importance sampling). The tests show that mod-Poisson approximation schemes lead to more accurate estimates for risk measures (such as $\mathrm{VaR}$ and $\mathrm{ES}$) and CDO tranche prices. Furthermore, due to the semi-analytical nature of the approximations, they require substantially less computational time, especially in the large portfolio limit. 
\clearpage 


\appendix

\section{Overview of estimation methods}
\label{sec:estimation_overview}

\subsection{Recursive methodology}
\label{subsec:recursive_methodology}

The recursive methodology was first introduced in \cite{barlow1984computing} (but see \cite{kuo2003optimal} for a more concise introduction) in the context of reliability theory, where the main quantity of interest is the failure probability of a system constituted by a large number of independent sub-components. The same estimation technique was rediscovered in the context of credit portfolios in \cite{hull2004valuation, brasch2004note} and is nowadays well--known among financial practitioners.
In its full generality, this method allows the exact computation of the law of $L_n = \sum_{i=1}^n Y_i$, where the random variables $(Y_i)_{i=1}^n$ are assumed to be integer-valued and independent, but not necessarily identically distributed. 
Under the conditional independence assumption typical of credit risk models, this method can be used to compute the conditional distribution of $L_n$ given $\Psi$, i.e.~$\P{L_n = k | \Psi}$ for all $k \in \Nb$ and all $\psi \in \range{\Psi}$. The unconditional distribution of $L_n$ can then be obtained via numerical integration over the mixing variable $\Psi$.
\medskip

Before presenting the recursive methodology algorithm in its full generality, it is instructive to consider the simpler case of the Poisson binomial distribution. Let $L_n = \sum_{i=1}^n Y_i$, where $(Y_i)_{i=1}^n$ are independent Bernoulli random variables, with $\P{Y_i = 1} = p_i$. The basic idea is to compute the distribution of the sum $L_n$ recursively by adding one Bernoulli random variable at a time. 

\vspace*{0.5cm}
\begin{algorithm}[ht]
\caption{Recursive methodology for loss distribution under assumption of independence.}
\label{alg:recursive_methodology}
\KwData{distributions of the $Y_i$ as $q_{i,j} = \P{Y_i = j}$, for $j \in \range{Y_i}$.}
\KwResult{distribution $p_k = \P{L_n = k}$, for $k \in \range{L_n}$.}
$p_0^{(0)} \gets 1$\\
$p_k^{(0)} \gets 0$, for $k \in \range{L_n} \backslash \{0\}$\\
\For{$i = 1$ \KwTo $n$}{
    $p_k^{(i)} \gets p_k^{(i-1)}$, for $k \in \range{L_n}$\\
    \For{$k \in \range{L_n}$}{
        \For{$j \in \range{Y_i}$}{
            $\ell \gets k + j$\\
            $p_k^{(i)} \minuseq p_k^{(i-1)} q_{i,j}$\\
            $p_\ell^{(i)} \pluseq p_k^{(i-1)} q_{i,j}$\\
        }
    }
}
\Return{$p_k \gets p_k^{(n)}$, for $k \in \range{L_n}$.}
\end{algorithm}
\vspace*{0.5cm}

If $L_n$ is the sum of the first term only, i.e.~$L_n = Y_1$, then its distribution is simply the same as that of $Y_1$. We can denote it as follows:
$$p^{(1)}_k = \P{L_n = k} = \begin{cases} 1 - p_1 & \mbox{if $k = 0$}, \\ p_1 & \mbox{if $k=1$} \end{cases}$$ 
where the superscript $(1)$ indicates that this is the distribution of $L_n$ as a sum of only one term.
If $L_n = L_1 + L_2$, then clearly we have that:
$$p^{(2)}_k = \P{L_n = k} = \begin{cases} (1 - p_1)(1-p_2) & \mbox{if $k = 0$}, \\ p_1(1-p_2) + p_2(1-p_1) & \mbox{if $k=1$}, \\ p_1 p_2 & \mbox{if $k=2$}, \end{cases}$$ 
which can also be expressed in terms of $p^{(1)}_k$ as follows:
$$p^{(2)}_k = p^{(1)}_k (1-p_2) + p^{(1)}_{k-1} p_2, \quad k = 0, 1, 2,$$
provided we set $p^{(1)}_{-1} := 0$.
This shows that adding one more summand to the sum is equivalent to shifting a probability mass equal to $p^{(1)}_{k} p_2$ from the point $k$ to the point $k+1$, for all possible values of $k$.
\medskip 

This recursion can be generalized, for any number of summands:
$$p^{(m)}_k = p^{(m-1)}_k (1-p_m) + p^{(m-1)}_{k-1} p_m, \quad k = 0, 1, \ldots, m.$$
Running the recursion from $m=1$ to $m=n$, one finally recovers the distribution of $L_n$.
In the general case with $L_n = \sum_{i=1}^n Y_i$, where the $Y_i$ are independent, integer-valued random variables not necessarily Bernoulli distributed, at the $m$-th step of the recursion we must shift a probability mass equal to $p_k^{(m-1)}\, \P{Y_m = j}$ from the point $k$ to the point $\ell := k + j$, for every possible value $j$ of $Y_m$ and for every possible value $k$. Algorithm \ref{alg:recursive_methodology} shows the recursion in this general case.
\medskip 

The complexity of the algorithm is
$$O\left(n \sum_{i=1}^{|\range{L_n}|} |\range{Y_i}| \right).$$
In the Poisson binomial case the complexity is $O(n^2)$, which makes the recursive methodology computationally expensive in the case of large portfolios, as shown in Section \ref{subsec:computational_time}.
We further remark that the recursive method always outputs the full loss distribution, which is computationally wasteful in the case of applications for which only a part of the distribution is needed, for instance its tail or a particular tail value.

\subsection{Large deviations approximation}

The large deviations theory provides tools for the analytical approximation of probabilities of rare events. Most results can be derived using a variety of techniques, but since the '80s a unified approach to the field has emerged, as expounded in several monographs on the subject \cite{deuschel2001large, ellis2006entropy, dembo2010large, den2008large}. 
The theory has been applied to many financial problems, including credit risk management (see \cite{pham2007some} for an overview). Here we follow the results presented in \cite{dembo2004large}, which can be used to estimate conditional probabilities of the form $\P{L_n > n x|\Psi}$ for a given tail value $x \in \Rb^+$ and a generic credit portfolio $L_n = \sum_{i=1}^n Z_i\, Y_i$.
\medskip

The key quantity, as in all large deviations applications, is the cumulant generating function of the random variable of interest. In this case we want to approximate the conditional distribution of $L_n$ given $\Psi$, so we denote its conditional cumulant generating function by $F(\cdot, \Psi)$ and compute it as:
\begin{equation}
\label{eq:large_devs_cgf}
F(\lambda, \Psi) = \frac{1}{n} \sum_{i=1}^n \log\left(1 + p_i(\Psi)(\E{\ee^{\lambda Z_i}|\Psi} - 1)\right),
\end{equation}
where $\P{Y_i = 1|\Psi} = p_i(\Psi)$ are the default probabilities as a function of the mixing variable and $\E{\ee^{\lambda Z_i}|\Psi}$ is the conditional moment generating function of $Z_i$. The latter quantity depends on the particular distributional assumptions of the credit risk model and should be known explicitly or be easily computable.

We remark that the normalization term $1/n$ in Equation \eqref{eq:large_devs_cgf} is necessary for adapting the proof of Bahadur--Rao's theorem (see \cite[Theorem 3.7.4.]{dembo2010large}) to the case of still independent, but not identically distributed summands. More specifically, $F(\lambda, \Psi)$ can be thought of as the cumulant generating function of a mixture of the summands' distributions, each taken with weight $1/n$.
\medskip

The second key quantity is the Legendre--Fenchel transform of $F(\lambda, \psi)$ for a given tail value $x$, which is defined as:
$$\lambda_x(\Psi) := \sup_{\lambda \ge 0} \left\{ \lambda x - F(\lambda, \Psi)\right\}.$$
Even in the case of the simplest credit models the value $\lambda_x(\psi)$ cannot be computed analytically, but fortunately the corresponding optimization problem can be solved efficiently, since the objective function is convex.
Finally, $F(\cdot, \Psi)$, its second derivative $F''(\cdot, \Psi)$ and $\lambda_x(\Psi)$ are used in the computation of the following large devations tail probability estimator:
\begin{equation}
\label{eq:large_deviations_estimator}
    \P{L_n > nx | \Psi} \approx \frac{1}{\sqrt{2 \pi n\, \lambda_x(\Psi)^2 F''(\lambda_x(\Psi), \Psi)}}\, \ee^{- n \left( \lambda_x(\Psi) x - F(\lambda_x(\Psi), \Psi) \right)}
\end{equation}
The unconditional tail probability can then be obtained numerically by integrating over the mixing variable $\Psi$. This latter integration actually suffers from numerical instabilities due to the presence of a vanishing denominator in Equation \eqref{eq:large_deviations_estimator} for some values of $\Psi$, as discussed in Section \ref{sec:risk_measures_estimation}.
We remark that the large deviations approximation is optimal only asymptotically in $n$, therefore we can expect it to perform well in the limit of very large portfolios.

\subsection{Stein's method: first-order correction to the Gaussian approximation}

In this section we present a method introduced by El Karoui and Jiao in \cite{el2009stein} and applied to CDO pricing in \cite{el2008gauss}. Their results rely on Stein's method \cite{stein1972bound} and on the zero-bias transformation framework \cite{goldstein1997stein} developed by Goldstein and Reinert. The results we are interested in rely on the following lemma.

\begin{lemma}
\label{lemma:gaussian_first_order_correction}
Let $X_1, \ldots, X_n \in \mathrm{L}^4$ be independent mean-zero random variables and let $W = X_1 + \cdots + X_n$, with $\sigma_W^2 = \mathrm{Var}(W)$. Then for any function $h$ such that $\|h''\|_\infty$ is finite, the following approximation holds:
$$ \E{h(W)} \approx \E{h(Z)} + \frac{\sum_{i=1}^n \E{(X_i)^3}}{2 \sigma_W^4}\, \E{\left( \frac{Z^2}{3 \sigma_W^2} - 1 \right) Z h(Z)},$$
where $Z \sim \Normal{0}{\sigma_W^2}$.
\end{lemma}

This lemma provides a first-order correction to the classical approximation of $W$ in terms of the Gaussian random variable $Z$. Its proof, together with an explicit bound on the approximation error, is beyond the scope of this work and can be found in \cite[Theorem 2.1]{el2008gauss}.
Unfortunately Lemma \ref{lemma:gaussian_first_order_correction} cannot be used directly for the estimation of tail probabilities or call prices, because the regularity assumption on the function $h$ (namely, that its second derivative exists everywhere) is not satisfied neither for indicator functions, $h(x) = \1_{\{x > k\}}$, nor for call functions, $h(x) = (x - k)^+$. Nevertheless the authors are able to show the result still holds for these functional choices: see \cite[Propositions 3.5 and 3.6]{el2009stein}.\medskip

If $L_n = \sum_{i=1}^n Z_i \,Y_i$ is a generic credit portfolio  with mixing variable $\Psi$, then we can approximate the expectation of a call function of $L_n$ with strike $K$ using Lemma \ref{lemma:gaussian_first_order_correction} applied to the centered random variables $X_i := Z_i Y_i - \E{Z_i Y_i}$, obtaining:
$$ \E{(L_n - K)^+|\Psi} \approx \E{(Z - K)^+|\Psi} + \frac{1}{6 \sigma_W^2} \sum_{i=1}^n \E{X_i^3} K \phi_{\sigma_W}(K),$$
where $\sigma_W^2 = \sum_{i=1}^n \mathrm{Var}(Z_i Y_i | \Psi)$, $Z \sim \Normal{0}{\sigma_W^2}$, $\phi_\sigma$ denotes the Gaussian density with variance $\sigma^2$, and the first-order correction term has been computed explicitly using Bachelier's formula.
The approximation for the unconditional expectation can then be obtained by numerical integration over the distribution of the mixing variable $\Psi$.

\subsection{Chen--Stein's method: first-order correction to the Poisson approximation}

This method was also introduced in \cite{el2009stein, el2008gauss}, together with the first-order Gaussian correction seen in the previous section. In fact the two methods share the same basic techniques, just different reference laws. When the reference law is the Poisson law, Stein's method can still be applied but for a different choice of operator, as first noticed by Chen in \cite{chen1975poisson}. 
The approximation we are interested in relies on the following lemma.

\begin{lemma}
\label{lemma:poisson_first_order_correction}
Let $X_1, \ldots, X_n \in \mathrm{L}^3$ be non-negative, integer-valued random variables and let $W = X_1 + \ldots + X_n$, with $\lambda_W = \E{W}$ and $\sigma_W^2 = \mathrm{Var}(W)$. Then for any bounded function $h$, the following approximation holds:
$$ \E{h(W)} \approx \E{h(Z)} + \frac{\sigma_W^2 - \lambda_W}{2}\, \E{\Delta^2_+ (h)(Z)},$$
where $Z \sim \Po{\lambda_W}$ and $\Delta^2_+(h)(x) = h(x+2) - 2 h(x+1) + h(x)$ is the second-order forward finite difference operator.
\end{lemma}

Also in this case, it is clear that the lemma provides a first-order correction to the classical Poisson approximation of $W$ in terms of $Z$. Its proof, together with an explicit bound on the approximation error, can be found in \cite[Theorem 4.2]{el2009stein}.\medskip 

In a credit risk setting we can apply this result to the case of a credit portfolio with unit exposures $L_n = \sum_{i=1}^n Y_i$ and mixing variable $\Psi$. For instance, we can approximate the expectation of a call function of $L_n$ with strike $K$ using Lemma \ref{lemma:poisson_first_order_correction} applied to the default indicators $Y_i \sim \Be{p_i(\Psi)}$, which yields:
$$ \E{(L_n - K)^+|\Psi} \approx \E{(Z - K)^+|\Psi} - \frac{1}{2}\sum_{i=1}^n p_i(\Psi)\, \ee^{-\lambda}\, \frac{\lambda^{K-1}}{(K-1)!},$$
where $\lambda = \sum_{i=1}^n p_i(\Psi)$ and $Z \sim \Po{\lambda}$.
The method can trivially be extended to the case of homogeneous deterministic recovery rates, that is portfolios of the form $L_n = \sum_{i=1}^n Z_i\, Y_i$, with $Z_i = 1 - R$, for some $R \in [0,1]$, but the authors in \cite{el2008gauss, el2009stein} are unable to provide an extension to the case of stochastic, heterogeneous recovery rates. In Section \ref{sec:mod-compound_poisson_approximation} we show how higher order correction terms to the Poisson approximation can be derived in this setting using mod-compound Poisson approximation schemes.
Finally, let us remark that this approximation method corresponds to the mod-Poisson approximation scheme of order $r=2$.

\subsection{Monte Carlo simulation}
\label{subsec:monte_carlo_simulation}

Simulation-based methods are very popular in financial applications and standard monographs in the field are \cite{asmussen2007stochastic} and \cite{glasserman2004monte}. 
The simplest simulation-based method is Monte Carlo integration. Algorithm \ref{alg:monte_carlo_simulation} presents the naive Monte Carlo estimator for the tail probabilities of a generic credit risk model $L_n = \sum_{i=1}^n L_i$. 
\medskip 

The simulation is divided in two steps: first we simulate a realization of the mixing variable $\Psi$ and then, conditionally on this value, we simulate the portfolio losses. The process is then repeated for a sufficiently high number of simulation runs.

\vspace*{0.5cm}
\begin{algorithm}[H]
\caption{Monte Carlo estimator for probability tail function.}
\label{alg:monte_carlo_simulation}
\KwData{number of Monte Carlo simulations $M$.}
\KwResult{tail probability estimator $\theta_x \approx \P{L_n > x}$, for some $x \in \Rb^+$.}
\For{$m = 1$ \KwTo $M$}{
    Sample $\psi_m$ from the distribution of $\Psi$\\
    \For{$i = 1$ \KwTo $n$}{
        Sample $L_i^{(m)}$ given $\Psi = \psi_m$\\
    }
}
$\displaystyle \theta_x \gets \frac{1}{M} \sum_{m=1}^M \1_{\left\{ \sum_{i=1}^n L_i^{(m)} > x \right\}}$\\
\Return{$\theta_x$.}
\end{algorithm}
\vspace*{0.5cm}

One key advantage of this naive Monte Carlo methodology is its generality: it's straightforward to introduce new sources of randomness and additional model parameters (provided, of course, the final losses can still be simulated efficiently), so that even very complicated models can be estimated.

Additionally, it is possible to estimate the full loss distribution with only one batch of simulations, by running Algorithm \ref{alg:monte_carlo_simulation} for different values of $x$ on the same simulated sample.

On the other hand, the naive Monte Carlo methodology requires at least $10^n$ simulations for the estimation of probabilities of order $10^{-n}$. This is a simple consequence of the fact that rare events occur rarely also in simulations, thus leading to estimators that underestimate the probability of rare events. In practice, even more simulations are needed in order to obtain an accurate estimate within any reasonable asymptotic confidence interval: see Figure \ref{fig:var_errors} for an empirical illustration. This problem is particularly severe in risk management applications, since large losses tend to occur with low probability but contribute substantially to the risk of a position.

\subsection{Importance sampling}
\label{subsec:importance_sampling}

In order to overcome the limitations of the naive Monte Carlo methodology for the estimation of rare events, importance sampling techniques have been developed. Our presentation follows \cite{glasserman2005importance}, where an importance sampling algorithm for the Gaussian copula model was first introduced, but good treatments can also be found in \cite{glasserman2004monte, mcneil2015quantitative}. Additionally we show explicitly how to incorporate stochastic exposures.
To illustrate the approach, let us first consider a simple Poisson binomial model of the form $L_n = \sum_{i=1}^n Y_i$ and suppose that we are interested in estimating the rare event probability $\P{L_n > x}$, for $x$ sufficiently large.
Let $f(y_1, \ldots, y_n) = \P{Y_1 = y_1, \ldots, Y_n = y_n}$ be the probability mass function of the random vector $(Y_1, \ldots, Y_n)$ on $\{0,1\}^n$ and let $g$ denote the probability mass function of another random vector of obligors' losses, $(\tilde{Y}_1, \ldots, \tilde{Y}_n)$, yet to be determined. 
Then a simple calculation shows that 
\begin{equation}
\label{eq:IS_estimator_formula}
  \P{L_n > x} = \E{\1_{\{Y_1 + \cdots+ Y_n > x\}}} = \E{\1_{\{\tilde{Y}_1 + \cdots + \tilde{Y}_n > x\}} \frac{f(\tilde{Y}_1, \ldots, \tilde{Y}_n)}{g(\tilde{Y}_1, \ldots, \tilde{Y}_n)}},  
\end{equation}
where the right-hand side integral is computed only with respect to the distribution of $(\tilde{Y}_1, \ldots, \tilde{Y}_n)$ and the ratio $f(\cdot)/g(\cdot)$ is known as the \emph{likelihood ratio function} of the two distributions.
\medskip 

The main idea of importance sampling is that we can choose $(\tilde{Y}_1, \ldots, \tilde{Y}_n)$ in such a way that the event $\{\tilde{Y}_1 + \cdots + \tilde{Y}_n > x \}$ is much more likely than $\{Y_1 + \cdots + Y_n > x\}$ on the support of the likelihood ratio function, so we can approximate the rare event probability on the left-hand side  of \eqref{eq:IS_estimator_formula} by estimating the integral on the right-hand side using a naive Monte Carlo estimation.
The large deviations theory suggests that a natural choice for the loss distribution $(\tilde{Y}_1, \ldots, \tilde{Y}_n)$ comes from exponentially tilting the total loss $L_n$, which leads to the following parametrized family:
\begin{align}
g_\lambda(y_1, \ldots, y_n) & = \ee^{\lambda (y_1 + \ldots + y_n)} \frac{f(y_1, \ldots, y_n)}{\E{\ee^{\lambda (Y_1 + \ldots + Y_n)}}} \label{eq:likelihood_ratio} \\
& = \ee^{\lambda (y_1 + \ldots + y_n)} \frac{\prod_{i=1}^n (p_i)^{y_i} (1-p_i)^{1-y_i}}{\prod_{i=1}^n (1 + p_i (\ee^\lambda - 1))} \nonumber \\
& = \prod_{i=1}^n {q_i(\lambda)}^{y_i} (1 - q_i(\lambda))^{1- y_i} \label{eq:exponential_tilting}
\end{align}
where we defined
$$ q_i(\lambda) = \frac{ p_i \ee^\lambda}{1 + p_i(\ee^\lambda - 1)}, \quad \lambda \in \Rb^+.$$
Equation \eqref{eq:exponential_tilting} shows that sampling from the distribution $g_\lambda$ is equivalent to simulating losses with new default probabilities $q_i(\lambda)$ instead of $p_i$.\medskip

Ideally the parameter $\lambda$ should be chosen by minimizing the variance of the importance sampling estimator in Equation \eqref{eq:IS_estimator_formula} (or equivalently its second moment), because the estimator with minimum variance will require the least number of simulations for any given level of confidence. Unfortunately this minimization is in general intractable, but by expressing the likelihood ratio function from Equation \eqref{eq:likelihood_ratio} in terms of $F(\lambda) := \log\left( \E{\ee^{\lambda (Y_1 + \ldots + Y_n)}} \right)$ (i.e.~the cumulant generating function of $L_n$), we obtain the following bound for the second moment of the exponentially tilted importance sampling estimator:
\begin{equation}
\label{eq:chernoff_bound}
\E{\left(\1_{\{\tilde{Y}_1 + \ldots + \tilde{Y}_n > x\}} \frac{f(\tilde{Y}_1, \ldots, \tilde{Y}_n)}{g_\lambda(\tilde{Y}_1, \ldots, \tilde{Y}_n)}\right)^2} \le \ee^{-2\lambda x + 2 F(\lambda)}.
\end{equation}
This bound, which is equivalent to the Chernoff bound, turns out to be remarkably sharp, so that the minimizer $\lambda_x$ of the right-hand side of \eqref{eq:chernoff_bound} can be used to obtain a very efficient importance sampling estimator.
The estimation procedure we just outlined applies to a Poisson binomial distribution, but it is readily extended to any Gaussian copula model of the form $L_n = \sum_{i=1}^n Z_i\, Y_i$ with mixing variable $\Psi$, assuming that the variables $Z_i \,Y_i$ are light-tailed.
\medskip

We first sample several realizations of the mixing variable $\Psi$ and compute the exponentially tilted importance sampling estimator for each realization. The presence of the stochastic exposures, $Z_i$, slightly modifies the formul{\ae} already presented through their moment generating functions, $\E{\ee^{\lambda Z_i}}$, which must be explicitly computable, but are otherwise easy to incorporate. A full description of this procedure can be found in Algorithm \ref{alg:importance_sampling_simulation_first_step}. 

\vspace*{0.5cm}
\begin{algorithm}[H]
\caption{One-step importance sampling algorithm for probability tail function.}
\label{alg:importance_sampling_simulation_first_step}
\KwData{number of Monte Carlo simulations $M$.}
\KwResult{tail probability estimator $\theta_x \approx \P{L_n > x}$, for some $x \in \Rb^+$.}
\For{$m = 1$ \KwTo $M$}{
    Sample $\psi_m$ from $\mathcal{N}(0,I)$\\
    Compute $\lambda_x(\psi_m)$\\
    \For{$i = 1$ \KwTo $n$}{
        $\displaystyle q_i \gets \frac{p_i(\psi_m) \E{\ee^{\lambda_x(\psi_m) Z_i}}}{1 + p_i(\psi_m)(\E{\ee^{\lambda_x(\psi_m) Z_i}} - 1)}$\\
        Sample $\tilde{L}_i^{(m)} = Z_i \tilde{Y}_i$, with $\tilde{Y}_i \sim \Be{q_i}$\\
    }
}
$\displaystyle \theta_x \gets \frac{1}{M} \sum_{m=1}^M \1_{\left\{ \sum_{i=1}^n \tilde{L}_i^{(m)} > x \right\}} \ee^{ -\lambda_x(\psi_m) \sum_{i=1}^n \tilde{L}_i^{(m)} + F(\lambda_x(\psi_m), \psi_m)}$\\
\Return{$\theta_x$.}
\end{algorithm}
\vspace*{0.5cm}

Nevertheless, this importance sampling procedure yields estimators that are far from being optimal, as shown in \cite{glasserman2005importance}. The problem lies in the way the mixing variable $\Psi$ is handled. Recall that in a Gaussian copula model large losses tend to occur for large realizations of the mixing variable $\Psi$ (up to sign conventions), but in Algorithm \ref{alg:importance_sampling_simulation_first_step} large realizations of $\Psi$ will be sampled only rarely. In other words, the exponential tilting must be applied to the unconditional portfolio losses and not just to the conditional ones.
This leads to a two-step exponential tilting procedure, in which the Gaussian mixing variable $\Psi$ is tilted first -- which amounts to a shift of its mean from zero to a new value $\mu$ -- while portfolio losses are tilted in a second step conditionally on each realization of the shifted mixing variable, exactly as in Algorithm \ref{alg:importance_sampling_simulation_first_step}. 
It should be mentioned that the first step relies heavily on the specific parametric choice of a Gaussian copula and furthermore on solving an approximate optimization, as the objective function itself needs to be approximated. Nevertheless several possible approximations are feasible and the interested reader is referred to \cite[Section 5.1]{glasserman2005importance} for an overview of choices.\medskip 

The final procedure is described in detail in Algorithm \ref{alg:importance_sampling_simulation_second_step} and this is also the procedure used for the empirical tests of Section \ref{sec:risk_measures_estimation}. Finally, we remark that if the full tail function of the loss distribution needs to be estimated, then it is not necessary to re-compute the shifted mean $\mu$ of $\Psi$ for each value of $x$, since in practice the same shift yields efficient estimators for a large neighborhood of tail points.
 
\vspace*{0.5cm}
\begin{algorithm}[H]
\caption{Two-step importance sampling algorithm for probability tail function.}
\label{alg:importance_sampling_simulation_second_step}
\KwData{number of Monte Carlo simulations $M$.}
\KwResult{tail probability estimator $\theta_x \approx \P{L_n > x}$, for some $x \in \Rb^+$.}
$\displaystyle \mu \gets \sup_{z \in \Rb} \left\{ F(\lambda_x(z), z) - \lambda_x(z) x - \frac{1}{2} |z|^2 \right\}$\\
\For{$m = 1$ \KwTo $M$}{
    Sample $\psi_m$ from $\mathcal{N}(\mu,I)$\\
    Compute $\lambda_x(\psi_m)$\\
    \For{$i = 1$ \KwTo $n$}{
        $\displaystyle q_i \gets \frac{p_i(\psi_m) \E{\ee^{\lambda_x(\psi_m) Z_i}}}{1 + p_i(\psi_m)(\E{\ee^{\lambda_x(\psi_m) Z_i}} - 1)}$\\
        Sample $\tilde{L}_i^{(m)} = Z_i \tilde{Y}_i$, with $\tilde{Y}_i \sim \Be{q_i}$\\
    }
}
$\displaystyle \theta_x \gets \frac{1}{M} \sum_{m=1}^M \1_{\left\{ \sum_{i=1}^n \tilde{L}_i^{(m)} > x \right\}} \ee^{ -\lambda_x(\psi_m) \sum_{i=1}^n \tilde{L}_i^{(m)} + F(\lambda_x(\psi_m), \psi_m) + \frac{1}{2} |\mu|^2 - \mu^T \psi_m}$\\
\Return{$\theta_x$.}
\end{algorithm}
\vspace*{0.5cm}

\section{Tail of the Poisson distribution}
\label{app:useful_formulae}

\begin{proposition}[Tail function of the Poisson distribution]
\label{prop:poisson_tail_gamma}

If $X$ is a $\Po{\lambda}$ random variable, then:
$$
    \P{X > k} = \frac{1}{k!}\, \gamma(k+1, \lambda), \quad \forall k \in \Nb
$$
where $\gamma$ is the lower incomplete gamma function given by:

$$ \gamma(x, \lambda) = \int_0^\lambda t^{x-1} \ee^{-t} dt.$$
\end{proposition}

\begin{proof}

For a Poisson distribution, the derivative of the tail probability with respect to the distribution parameter $\lambda$ is given by the probability mass function. This can be shown as follows:

\begin{align*}
    \frac{\dd}{\dd \lambda} \P{X > k} & = \sum_{j = k+1}^\infty \frac{1}{j!} \frac{\dd}{\dd \lambda} \left( \ee^{-\lambda} \lambda^j \right) \\
    & = \sum_{j=k+1}^\infty \ee^{-\lambda} \frac{\lambda^{j-1}}{(j-1)!} - \sum_{j=k+1}^\infty \ee^{-\lambda} \frac{\lambda^{j}}{(j)!} \\
    & = \sum_{j=k}^\infty \ee^{-\lambda} \frac{\lambda^{j}}{j!} - \sum_{j=k+1}^\infty \ee^{-\lambda} \frac{\lambda^{j}}{(j)!} = \ee^{-\lambda} \frac{\lambda^k}{k!}.
\end{align*}

Then, by the fundamental theorem of calculus, one has:

\begin{align*}
    \P{X > k} 
    & = \int_0^\lambda \ee^{-t}\, \frac{t^k}{k!}\, dt  = \frac{1}{k!}\, \gamma(k+1, \lambda).
\end{align*} 
\end{proof}



\section{Incidence algebras and the M\"obius function}
\label{sec:mobius_function}

\begin{definition}
A poset (or partially ordered set) $P$ is a set together with a binary order relation, denoted $\le$, satisfying the following axioms:
\begin{enumerate}
    \item $x \le x, \quad \forall x \in P$,
    \item if $x \le y$ and $y \le x$, then $x = y$,
    \item if $x \le y$ and $y \le z$, then $x \le z$.
\end{enumerate}
\end{definition}

We say that $P$ has a minimal element, denoted $\widehat{0}$, if there exists an element $\widehat{0} \in P$ such that $\widehat{0} \le x$ for all $x \in P$. Analogously, has a maximal element $\widehat{1}$, if there exists an element $\widehat{1} \in P$ such that $\widehat{1} \ge x$ for all $x \in P$.
 
\begin{example}[The poset of set partitions $\Pi(n)$]
Let $A$ be a finite set. A set partition $\pi = \{B_1, \ldots, B_k\}$ of $A$ is a collection of non-empty, mutually disjoint subsets of $A$, such that $\cup_{i=1}^k B_i = A$. The sets $B_1, \ldots, B_k$ are called the blocks of $\pi$ and the number of blocks of $\pi$ is denoted by $|\pi|$.
\medskip

For $n \in \Nb$, define $\Pi(n)$ as the set of all set partitions of $\{1, 2, \ldots, n\}$.  Given two set partitions $\pi$ and $\sigma$, we denote $\pi \le \sigma$ if every block of $\pi$ is contained in a block of $\sigma$. Then, the poset $(\Pi(n),\le)$
admits a minimal element $$\widehat{0}_n = \{\{k\}, k = 1, \ldots, n\}\quad\text{(the partition with $n$ blocks)}$$ and a maximal element 
$$\widehat{1}_n = \{\{1, \ldots, n\}\}\quad\text{(the partition with only one block)}.$$
\end{example}

An interval of a poset $P$, denoted $[x, y]$ for some $x, y \in P$ with $x \le y$, is a subset of $P$ defined as $[x, y] = \{ z \in P \: | \: x \le z \le y\}$. We denote by $\mathrm{Int}(P)$ the set of all intervals of $P$ and we say that $P$ is locally finite is every interval of $P$ is finite.
It is quite natural to define functions on intervals, for instance if we want to count the number of elements of an interval, and more generally functions $f:\mathrm{Int}(P) \to \Kb$ for some field $\Kb$. The space of all such functions can be turned into an associative algebra, as the following definition shows.

\begin{definition}[Incidence algebra $I(P, \Kb)$]
The incidence algebra $I(P, \Kb)$ of $P$ over $\Kb$ is the $\Kb$-algebra of all functions $f:\mathrm{Int}(P) \to \Kb$ with operation, called convolution, given by:
$$ (f \star g)(x,y) := \sum_{x \le z \le y} f(x, z)\, g(z, y).$$
\end{definition}

The algebra $I(P, \Kb)$ is associative, and its multiplicative identity is
$$ \delta(x, y) = \begin{cases} 1 & \text{if $x = y$,} \\ 0 & \text{otherwise.} \end{cases}$$
Another important element of the incidence algebra is the zeta function of the poset, defined as
$$ \zeta(x, y) = \begin{cases} 1 & \text{if $x \le y$,} \\ 0 & \text{otherwise.} \end{cases}$$
It can be shown that the function $\zeta$ of a poset is invertible and its inverse is called the M\"obius function of the poset and is denoted by $\mu$. In particular one has that $\mu$ satisfies:
$$ \mu \star \zeta = \zeta \star \mu = \delta.$$
Furthermore, the following important result holds.

\begin{theorem}[M\"obius inversion formula]
Let $P$ be a locally finite poset and let $g,h: P \to \Kb$. Then
$$ h(x) = \sum_{y \le x} g(y), \quad \forall x \in P $$
is equivalent to
$$ g(x) = \sum_{y \le x} h(y)\, \mu(y, x), \quad \forall x \in P. $$
\end{theorem}

\begin{example}
The M\"obius function of the poset of set partitions $\Pi(n)$ admits the following representation:
$$ \mu(\pi, \sigma) = (-1)^{|\pi| - |\sigma|} \prod_{B \in \sigma} (n^\sigma_\pi(B) - 1)!, \quad \forall \pi \le \sigma$$
where $n^\sigma_\pi(B)$ is the number of blocks of $\pi$ contained in the block $B$ of $\sigma$.
\end{example}

We refer to \cite{rota1964} for details on these constructions. For the manipulation of Fourier and Laplace transforms of probability distributions, the formalism of posets and Möbius functions enables one to go from a generating series to its exponential or logarithm.
\begin{theorem}[Exponential and logarithm of generating series]\label{thm:exp_log}
Let $G(z) = \sum_{n=1}^\infty \frac{g_n}{n!}\,z^n$ be the exponential generating series of a sequence of coefficients $(g_n)_{n \geq 1}$. If $H(z)=\exp(G(z))$, then
$$H(z) = 1 + \sum_{n=1}^\infty \frac{1}{n!}\left(\sum_{\pi \in \Pi(n)} \prod_{B \in \pi} g_{|B|} \right) z^n.$$
Conversely, if $H(z) = 1 + \sum_{n=1}^\infty \frac{h_n}{n!}\,z^n$ and $G(z) = \log(H(z))$, then
$$G(z) = \sum_{n=1}^\infty \frac{1}{n!}\left(\sum_{\pi \in \Pi(n)} \mu(\pi,\widehat{1}_n)\,\prod_{B \in \pi} h_{|B|} \right)z^n,$$
with $\mu(\pi,\widehat{1}^n) = (-1)^{|\pi|-1}\,(|\pi|-1)!$.
\end{theorem}

\begin{proof}
We expand the exponential of $G(z) = \sum_{n=1}^\infty \frac{g_n}{n!}\,z^n$, and we collect the coefficient of $z^n$. This is
$$ [z^n]\,H(z) = \sum_{l =1}^n  \sum_{\substack{c_1+\cdots+c_l=n \\ c_1\geq 1,\ldots,c_l\geq 1}} \frac{1}{l!} \frac{g_{c_1}\cdots g_{c_k}}{c_1!\cdots c_l!} .$$
The sum above runs over compositions of $n$, that is to say sequences $(c_1,\ldots,c_l)$ of positive integers with sum equal to $n$. By replacing these \emph{compositions} of size $n$ by their non-increasing reorderings called \emph{integer partitions} of size $n$, we obtain a sum over a smaller set:
$$ [z^n]\,H(z) = \sum_{l =1}^n  \sum_{\substack{\lambda_1+\cdots+\lambda_l=n \\ \lambda_1\geq \cdots \geq \lambda_l\geq 1}} \frac{1}{m_1(\lambda)!\cdots m_n(\lambda)!} \frac{g_{\lambda_1}\cdots g_{\lambda_l}}{\lambda_1!\cdots \lambda_l!} ,$$
where $m_i(\lambda)$ denotes the number of parts $\lambda_{1\leq j\leq l}$ of $\lambda$ equal to $i$. Indeed, given an integer partition $\lambda=(\lambda_1\geq \cdots \geq \lambda_l)$ with sum $n$, the number of compositions whose non-increasing reordering is $\lambda$ is the multinomial coefficient $\frac{l!}{m_1(\lambda)!\cdots m_n(\lambda)!}$. Now, for any integer partition $\lambda$ with size $n$,
$$\frac{n!}{m_1(\lambda)!\cdots m_n(\lambda)!\,\lambda_1!\cdots \lambda_l!}$$
is the number of \emph{set partitions} $\pi$ with size $n$ and type $\lambda$, that is to say that the sizes of the blocks of $\pi$ are given by the integer partition $\lambda$. Therefore,
$$n!\,[z^n]\, H(z) = \sum_{\pi \in \Pi(n)} \prod_{B \in \pi} g_{|B|}.$$
This proves the first formula. In order to get the second formula, let us define two functions $g$ and $h$ on the poset $\Pi(n)$ :
\begin{align*}
    g(\pi) &= \prod_{B \in \pi} \big(|B|!\,[z^{|B|}]\,G(z)\big) = \prod_{B \in \pi} g_{|B|} ;\\ 
    h(\pi) &= \prod_{B \in \pi} \big(|B|!\,[z^{|B|}]\,H(z)\big) = \prod_{B \in \pi} h_{|B|}.
\end{align*}
If $H(z)=\exp(G(z))$, then we have shown that $h(\widehat{1}_n) = \sum_{\pi \in \Pi(n)} g(\pi)$, from which we deduce that $h(\sigma) = \sum_{\pi \leq \sigma} g(\pi)$. By Möbius inversion, $g(\sigma) = \sum_{\pi \leq \sigma} \mu(\pi,\sigma)\,g(\pi)$, so in particular,
$$g(\widehat{1}_n) = g_n= n!\,[z^n]\,G(z) = \sum_{\pi \in \Pi(n)} \mu(\pi,\widehat{1}_n)\,\prod_{B \in \pi} h_{|B|}.$$
Note that an alternative way to get these inversion formul{\ae} is by means of the Faà-di-Bruno formula
\begin{equation*}
\frac{\dd^n(a\circ b)}{\dd z^n} (z)= \sum_{\pi \in \Pi(n)} \frac{\dd^{|\pi|}a}{\dd z^{|\pi|}} (b(z)) \prod_{B \in \pi} \frac{\dd^{|B|}b}{\dd z^{|B|}} (z),    
\end{equation*}
with $b$ taken equal to $G(z)$ or $H(z)$, and $a=\exp$ or $a=\log$. 
\end{proof}

\section{Relation between the coefficients of the approximation scheme and the moments of the total loss variable}\label{sec:stirling}
In this appendix, we prove Formula \eqref{eq:bkn_in_terms_of_Mkn}. It is convenient to introduce the elementary symmetric functions
$$\mathfrak{e}_{k,n} = \sum_{1\leq i_1<i_2 < \cdots < i_k\leq n}p_{i_1}p_{i_2}\cdots p_{i_k}.$$
By using the well known relation between the coefficients of a polynomial and its roots, we get:
\begin{align*}
    1+\sum_{k=1}^\infty \mathfrak{e}_{k,n}\,z^k &= \prod_{i=1}^n (1+p_iz) = \exp\left(\sum_{i=1}^n \log(1+p_iz)\right) \\ 
    &= \exp\left(\sum_{i=1}^n\sum_{k=1}^\infty \frac{(-1)^{k-1}}{k}\,(p_iz)^k\right) =\exp\left(\sum_{k=1}^\infty \frac{(-1)^{k-1}}{k}\,\mathfrak{p}_{k,n}\,z^k\right).
\end{align*}
By using Theorem \ref{thm:exp_log}, we then obtain the relation between the coefficients $\mathfrak{e}_{k,n}$ and the coefficients $\mathfrak{p}_{k,n}$:
$$\mathfrak{e}_{k,n} = \frac{1}{k!} \sum_{\pi \in \Pi(k)} \mu(\widehat{0}_k,\pi)\left(\prod_{B \in \pi} \mathfrak{p}_{|B|,n}\right).$$
This is the same relation as Equation \eqref{eq:mod-poisson_approximation_scheme_coefs}, except that the sum runs over all set partitions, and not only those with blocks of size larger than $2$. As a consequence, the coefficients $\mathfrak{e}_{k,n}$ and $b_{k,n}$ are related by the following inclusion-exclusion formula:
\begin{equation}
    b_{k,n} = \sum_{l=0}^k \frac{(-1)^l}{l!}\,(\mathfrak{e}_{1,n})^l\,\mathfrak{e}_{k-l,n}.\label{eq:inclusion_exclusion}
\end{equation}
Indeed, let us replace on the right-hand side each $\mathfrak{e}_{k-l,n}$ by its expansion over set partitions. We get:
$$\mathrm{RHS} = \sum_{l=0}^k \sum_{\pi \in \Pi(k-l) } \frac{(-1)^l\,\mu(\widehat{0}_{k-l},\pi)}{l!\,(k-l)!}\,(\mathfrak{p}_{1,n})^l\left(\prod_{B \in \pi} \mathfrak{p}_{|B|,n}\right).$$
Since there are $\frac{k!}{l!(k-l)!}$ subsets of the integer interval $\lle 1,k\rre=\{1,2,\ldots,k\}$ with size $l$, we can rewrite the formula above as a sum over pairs $(\sigma,S)$, where $\sigma \in \Pi(k)$, and $S$ is a subset of $S(\sigma)=\bigsqcup_{B \in \sigma,\,|B|=1} B$, which is the union of the blocks of $\sigma$ with size $1$. Thus,
$$\mathrm{RHS} = \frac{1}{k!}\sum_{\sigma \in \Pi(k)} \sum_{S \subset S(\sigma)} (-1)^{|S|}\, \mu(\widehat{0}_{k},\sigma)\left(\prod_{B \in \sigma} \mathfrak{p}_{|B|,n}\right).$$
Given a set partition $\sigma$, the alternate sum $\sum_{S \subset S(\sigma)} (-1)^{|S|}$ vanishes unless $S(\sigma)=\emptyset$, so we conclude that
$$\mathrm{RHS} = \frac{1}{k!}\sum_{\substack{\sigma \in \Pi(k) \\ \forall B \in \sigma,\,|B|\geq 2}} \mu(\widehat{0}^k,\sigma)\,\left(\prod_{B \in \sigma} \mathfrak{p}_{|B|,n}\right)=b_{k,n}.$$
Now, the coefficients $\mathfrak{e}_{k,n}$ also appear in the computations of the moments of $L_n = \sum_{i=1}^n Y_i$ with $Y_i \sim \Be{p_i}$:
\begin{align*}
    \E{(L_n)^r} &= \sum_{1\leq i_1,i_2,\ldots,i_r \leq n} \E{Y_{i_1}Y_{i_2}\cdots Y_{i_r}} = \sum_{f : \lle 1,r\rre \to \lle 1,n\rre} \left(\prod_{i \in f(\lle 1,r\rre)} p_i\right) \\
    &= \sum_{s=1}^r \sum_{1\leq i_1<i_2<\cdots <i_s\leq n } s!\,\genfrac\{\}{0pt}{0}{r}{s}\, p_{i_1}p_{i_2}\cdots p_{i_s} =  \sum_{s=1}^r s!\,\genfrac\{\}{0pt}{0}{r}{s}\, \mathfrak{e}_{s,n},
\end{align*}
where $\genfrac\{\}{0pt}{1}{r}{s}$ is the Stirling number of the second kind, which counts set partitions of $\lle 1,r\rre$ in $s$ parts. Indeed, to go from the first line to the second line, we gather the functions $f:\lle 1,r\rre \to \lle 1,n\rre$ according to their range $f(\lle 1,r\rre) = \{i_1<i_2<\cdots<i_s\}$; if this range is fixed, then there are $s!\, \genfrac\{\}{0pt}{1}{r}{s}$ functions with this range. Set $M_{r,n}=\E{(L_n)^r}$. The relation above can be inverted by introducing the Stirling number of the first kind  $\genfrac[]{0pt}{1}{r}{s}$, which counts permutations of size $r$ with $s$ disjoint cycles. Hence,
\begin{equation}
    \mathfrak{e}_{r,n} = \frac{1}{r!}\,\sum_{s=1}^r (-1)^{r-s}\, \genfrac[]{0pt}{0}{r}{s} \,M_{s,n};\label{eq:stirling_numbers}
\end{equation}
see \cite[Sections 1.3 and 1.4]{Stan97} for the combinatorial properties of the two kinds of Stirling numbers. Injecting Equation \eqref{eq:stirling_numbers} into Formula \eqref{eq:inclusion_exclusion}, we get Equation \eqref{eq:bkn_in_terms_of_Mkn}.











\clearpage
\bibliographystyle{amsalpha}
\bibliography{credit_risk}
\end{document}